\documentclass[AMS,Times1COL]{WileyNJDv5}

\articletype{Article Type}%

\received{Date Month Year}
\revised{Date Month Year}
\accepted{Date Month Year}
\journal{Journal}
\volume{00}
\copyyear{2025}
\startpage{1}

\usepackage{pgfplots}
\newcommand{\x}{x}
\newcommand{\p}{p}
\newcommand{\w}{w}

\renewcommand{\u}{u}
\newcommand{\q}{q}
\newcommand{\z}{z}
\newcommand{\zs}[1][]{{\boldsymbol{z}^{#1}}}
\newcommand{\y}{y}
\newcommand{\filty}{s}
\renewcommand{\t}{t}
\renewcommand{\k}{k}

\newcommand{\A}[1]{A_{#1}}
\newcommand{\B}[2]{B_{#1}^{#2}}
\newcommand{\C}[2]{C_{#1}^{#2}}
\newcommand{\D}[3]{D_{#1}^{#2#3}}
\newcommand{\BB}[1]{B_{#1}}
\newcommand{\CC}[1]{C_{#1}}
\newcommand{\DD}[1]{D_{#1}}

\newcommand{\AK}{A_\K}
\newcommand{\BK}{B_\K}
\newcommand{\CK}{C_\K}
\newcommand{\DK}{D_\K}
\newcommand{\G}{G}                            
\newcommand{\Grho}{{\G_\rho}}
\newcommand{\K}{K}                            
\newcommand{\Kx}{\kappa}                      
\newcommand{\Krho}{{\K_\rho}}

\newcommand{\GK}{{\Theta}}                   
\newcommand{\GKrho}{{\GK_\rho}}
\newcommand{\DG}{\Delta \star \G}
\newcommand{\DGK}{\Delta\star \G \star \K}
\newcommand{\filt}{\Psi}                      
\newcommand{\filtx}{\psi}                     
\newcommand{\all}{\Sigma}                     
\newcommand{\allrho}{{\Sigma_\rho}}
\newcommand{\allhat}{{\hat{\all}}}
\newcommand{\allhatrho}{{\hat{\all}_\rho}}
\newcommand{\allinv}{\Omega}                  
\newcommand{\allinvrho}{{\Omega_\rho}}
\newcommand{\genplant}{{\mathcal{G}}}                
\newcommand{\genplantrho}{{\genplant_\rho}}

\newcommand{\DGKrho}{\Delta_\rho \star \G_\rho \star \K_\rho }

\newcommand\Psyn{\mathcal{\P}}
\newcommand\Xsyn{\mathcal{\P}_{22}}
\newcommand\Ysyn{\mathcal{\P}_{11}}
\newcommand\Ksyn{\mathcal{K}}
\newcommand\Lsyn{\mathcal{L}}
\newcommand\Msyn{\mathcal{M}}
\newcommand\Nsyn{\mathcal{N}}

\newcommand\Asyn{\mathcal{A}_\rho}
\newcommand\Bsyn[1]{\mathcal{B}^{#1}_\rho}
\newcommand\Csyn[1]{\mathcal{C}^{#1}}
\newcommand\Dsyn[2]{\mathcal{D}^{{#1}{#2}}}

\newcommand\Usyn{\mathcal{U}}
\newcommand\Vsyn{\mathcal{V}}
\newcommand\Ysyntrafo{\mathcal{Y}}
\DeclareMathOperator{\diag}{diag}

\newcommand{\symb}[1]{\begin{pmatrix}\vphantom{#1}\star \end{pmatrix}^{\hspace{-0.15cm}\top}\hspace{-0.1cm} #1 }
\newcommand{\symbscalar}{(\star)^{\hspace{-0.05cm}\top}\hspace{-0.025cm} }
\newcommand{\symh}[1]{\begin{pmatrix}\vphantom{#1}\star \end{pmatrix}^{\!*}\! #1 }
\newcommand{\diagmat}[1]{{\arraycolsep=2pt\begin{pmatrix} #1\end{pmatrix}}}
\newcommand{\Z}{Z}

\newcommand{\Zdone}{\mathcal X_{1}}
\newcommand{\Zdtwo}{\mathcal X_{2}}

\newcommand{\ZX}{W}
\newcommand{\ZY}{Y}
\newcommand{\nx}{{n_\x}}
\let\greeknu\nu
\renewcommand{\nu}{{n_\u}}
\newcommand{\nw}{{n_\w}}

\newcommand{\nq}{{n_\q}}
\newcommand{\np}{{n_\p}}
\newcommand{\ny}{{n_\y}}
\newcommand{\nz}{{n_\z}}
\newcommand{\nK}{{n_\Kx}}

\newcommand{\nfilty}{{n_\filty}}

\newcommand{\nfilt}{{n_\filtx}}
\newcommand{\nfiltone}{{n_{\filtx_1}}}

\newcommand{\nfilttwohat}{{n_{\hat\filtx_2}}}
\newcommand{\nfilthat}{{n_{\hat\filtx}}}

\newcommand{\filtyone}{{\filty_1}}

\newcommand{\M}{M}
\newcommand{\X}{X}
\renewcommand{\P}{P}
\newcommand{\Pschur}{\mathscr{P}}
\newcommand{\T}{T}

\DeclareSymbolFont{newfont}{OML}{cmm}{m}{it}
\newcommand{\rhosq}{\rho^2}

\newcommand{\R}{\mathbb{R}}
\newcommand{\N}{\mathbb{N}}

\renewcommand{\S}{\mathbb{S}}
\newcommand{\Hinf}{\mathcal{H}_\infty}
\newcommand{\RHinf}{\mathcal{RH}_\infty}

\newcommand{\blkmat}[2]{\left(\begin{array}{@{}#1@{}} #2 \end{array}\right)}
\newcommand{\newblkdash}[1][2.25ex]{\\[0.02cm] \hdashline\rule{0pt}{#1}}
\newcommand{\newblk}[1][2.25ex]{\\[0.02cm] \hline\rule{0pt}{#1}}
\newcommand{\ssrep}[4]{\left[\begin{array}{c|c} #1 & #2 \newblk #3 & #4 \end{array}\right]}
\newcommand{\ssrepflex}[4]{\left[\begin{array}{#1|#2} #3 \newblk #4 \end{array}\right]}

\newcommand{\peak}{{\mathrm{peak}}}
\newcommand{\ETE}{\infty}
\usepackage{stmaryrd}
\newcommand{\PTP}{{\textnormal{p}{\ensuremath{\shortrightarrow }}\textnormal{p}}}
\newcommand{\ETP}{{2{\ensuremath{\shortrightarrow}}\textnormal{p}}}
\newcommand{\Deltaset}{\boldsymbol{\Delta}}

\newcommand{\MXset}{\mathbb{\M\X}}

\DeclareMathOperator{\dare}{dare}
\DeclareMathOperator*{\argmin}{arg\,min}

\usepackage{accents}
\newcommand{\ubar}[1]{\underaccent{\bar}{#1}}

\usepackage{xfrac}
\usepackage{tikz}
\usepackage{algpseudocode}
\usetikzlibrary{positioning}
\usetikzlibrary{calc}
\usetikzlibrary{decorations.markings,patterns,decorations.pathmorphing}
\usetikzlibrary{arrows, shapes}
\tikzstyle{block} = [draw, thick, node distance=0.5cm, minimum width=1cm, inner sep=6pt]
\tikzstyle{sum} = [draw, thick, circle, node distance=1cm, inner sep=3.5pt, path picture={\node at (path picture bounding box.center) [draw, anchor = center] {$+$};}]

\newcommand{\refeq}[2]{\overset{\makebox[0pt][c]{\small #1}}{#2}}
\theoremstyle{remark}
\newtheorem{example}{Example}
\usepackage{arydshln}
\usepackage{mathrsfs}  

\usepackage{color}
\definecolor{myblue}{rgb}{0.00000,0.44700,0.74100}%
\definecolor{myred}{rgb}{0.95000,0.200,0.200}%
\definecolor{mycolor3}{rgb}{0.92900,0.69400,0.12500}%
\definecolor{mygreen}{rgb}{0.1,0.7,0.15}%
\hypersetup{
	colorlinks   = true, 
	urlcolor     = myblue, 
	linkcolor    = myblue, 
	citecolor   = mygreen 
}

\renewcommand{\doi}[1]{\href{https://doi.org/#1}{#1}}

\usepackage{dsfont}
\newcommand{\eins}{\mathds{1}}
\newcommand{\mytitle}{Multi-objective robust controller synthesis with integral quadratic constraints in discrete-time}
\renewcommand{\thefootnote}{\arabic{footnote}}

\begin{document}

\title{\mytitle}

\author[1]{Lukas Schwenkel}

\author[2]{Johannes Köhler}

\author[3]{Matthias A. Müller}

\author[4]{Carsten W. Scherer}

\author[1]{Frank Allgöwer}

\authormark{L. Schwenkel, J. Köhler, M. A. Müller, C. W. Scherer, F. Allgöwer}
\titlemark{\mytitle}

\address[1]{\orgdiv{University of Stuttgart}, \orgname{Institute for Systems Theory and Automatic Control}, \orgaddress{\country{Germany}}}

\address[2]{\orgdiv{ETH Zürich}, \orgname{Institute for Dynamical Systems and Control}, \orgaddress{\country{Switzerland}}}

\address[3]{\orgdiv{Leibniz University Hannover}, \orgname{Institute for Automatic Control}, \orgaddress{\country{Germany}}}

\address[4]{\orgdiv{University of Stuttgart}, \orgname{Department of Mathematics}, \orgaddress{\country{Germany}}}

\corres{Corresponding author Lukas Schwenkel \email{lukas.schwenkel@ist.uni-stuttgart.de}}

\abstract[Abstract]{This article presents a novel framework for the robust controller synthesis problem in discrete-time systems using dynamic Integral Quadratic Constraints (IQCs). 
We present an algorithm to minimize closed-loop performance measures such as the $\Hinf$-norm, the energy-to-peak gain, the peak-to-peak gain, or a multi-objective mix thereof.
While IQCs provide a powerful tool for modeling structured uncertainties and nonlinearities, existing synthesis methods are limited to the $\Hinf$-norm, continuous-time systems, or special system structures. 
By minimizing the energy-to-peak and peak-to-peak gain, the proposed synthesis can be utilized to bound the peak of the output, which is crucial in many applications requiring robust constraint satisfaction, input-to-state stability, reachability analysis, or other pointwise-in-time bounds.
Numerical examples demonstrate the robustness and performance of the controllers synthesized with the proposed algorithm.}

\keywords{robust control, uncertain systems, integral quadratic constraints, multi-objective control design}

\jnlcitation{\cname{%
\author{L. Schwenkel},
\author{J. Köhler},
\author{M. A. Müller},
\author{C. W. Scherer}, and
\author{F. Allgöwer}}.
\ctitle{\mytitle} 
\cjournal{\it arXiv preprint.}
}

\maketitle

\renewcommand\thefootnote{}
\footnotetext{\textbf{Abbreviations:} IQC, integral quadratic constraint; SISO, single-input-single-output; LMI, linear matrix inequality}

\renewcommand\thefootnote{\fnsymbol{footnote}}
\setcounter{footnote}{1}

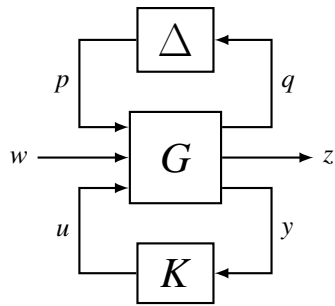
\begin{figure}[tb]\centering
    \begin{tikzpicture}
      \node [block] (Delta) {\Large$\Delta$};
      \draw[white] (Delta.west) + (-0.8,0) -- +(-1.2,0) node[left]{$\p$};
      \node [block, below=of Delta, minimum height=1.2cm, minimum width=1.2cm] (G) {\Large $\G$};
      \draw[latex-, thick] (G.west) -- +(-1.2,0)node [left] {$\w$};
      \node [block, below = of G] (K) {\Large $K$};
      \draw[-latex, thick] (G.east)+(0,-0.4)  --  +(0.65,-0.4) |-node [right,pos=0.25] {$\y$} (K);
      \draw[latex-, thick] (G.west)+(0,-0.4)  -- +(-0.65,-0.4)|- node [left,pos=0.25] {$\u$} (K);
      
      \draw[-latex, thick] (G.east) +(0,0) -- ($(G.east)+(1.2,0)$) node [right] {$\z$};
      \draw[latex-, thick] (Delta.east) -| node[right, pos=0.75] {$\q$} ($(G.east)+(0.65,0.4)$) -- ($(G.east)+(0,0.4)$);
      \draw[-latex, thick] (Delta.west) -| node[left, pos=0.75] {$\p$} ($(G.west)+(-0.65,0.4)$) -- ($(G.west)+(0,0.4)$);
    \end{tikzpicture}
  \caption{Interconnection of the system $\G$, the uncertainty $\Delta$, and the controller $\K$ with the performance channel from $\w$ to $\z$.}\label{fig:DGK}
  \end{figure}

\section{Introduction}
This article introduces a systematic procedure to design discrete-time robust controllers for uncertain systems of the form in Figure~\ref{fig:DGK}. 
Our goal is to minimize the worst-case $\mathcal H_\infty$-norm, energy-to-peak gain, or peak-to-peak gain of the closed loop.
For single-input-single-output (SISO) systems this includes $\mathcal H_2$-performance (equal to the energy-to-peak gain) and $\ell_1$-performance (equal to the peak-to-peak gain) opening a vast field of possible applications of our results.
Furthermore, we show how to optimize a mix of these performance criteria in a multi-objective setting.
We consider a large class of uncertainties by characterizing them using integral quadratic constraints (IQCs).
IQCs provide a unified approach to model a wide range of structured uncertainties and nonlinearities within a dynamical system~\cite{Megretski1997, Veenman2016}.

\emph{Contribution.} 
First, we derive a new robust analysis result to compute upper bounds on the worst-case $\mathcal H_\infty$-norm, the energy-to-peak, and the peak-to-peak gain based on linear matrix inequalities (LMIs).
Second, our main contribution is an algorithm that synthesizes discrete-time robust controllers which minimize these performance measures.
Due to the time-domain nature of the peak norm, we use dissipativity arguments as in~\cite{Hu2016,Scherer2022a,Scherer2018,Seiler2015} and compute factorizations in the state space as in~\cite{Veenman2014}.
In particular, we characterize uncertainty using finite horizon IQCs with terminal cost as introduced in~\cite{Scherer2018} which constitute a seamless integration of classic frequency-domain IQCs~\cite{Megretski1997} into the time domain.
The proposed algorithm can further be used to optimize a mix of performance criteria for different channels by using a common Lyapunov matrix (as known from the nominal multi-objective case~\cite{Scherer1997}) and a common IQC multiplier in the analysis and synthesis steps.
We show that the performance bounds provided in each iteration step of our algorithm are non-increasing by explicitly constructing feasible warm starts which can also be used to speed up the optimization.

\emph{Applications of discrete-time IQCs.} 
There has been an increasing interest in the past decade in the discrete-time IQC framework~\cite{Fetzer2017,Hu2017,Kao2012} with applications to the analysis~\cite{Fazlyab2018,Jakob2025,Lessard2016} and design~\cite{Lessard2020,Michalowsky2020,Scherer2021} of optimization algorithms, the analysis of neural network interconnections~\cite{Pauli2021,Yin2020a}, the synthesis of neural network controllers~\cite{Gu2022}, as well as in the context of model predictive control to construct stabilizing terminal conditions~\cite{Morato2023} or compute a reachable set~\cite{Schwenkel2020,Schwenkel2022a}.
The discrete-time IQC synthesis results~\cite{Gu2022,Lessard2020,Michalowsky2020,Morato2023,Scherer2021} exploit specific properties of the considered applications and hence are limited to a particular setting.
The general approach presented in this article allows for extensions of these works to more uncertainty classes, tighter IQC descriptions, and other performance criteria. 
One particular application of the proposed peak-to-peak synthesis procedure is to design the pre-stabilizing controller for robust model predictive control~\cite{Schwenkel2020,Schwenkel2022a} such that the volume of the reachable set is minimized.

\emph{Related work on IQC analysis.} 
A summary of results using IQCs to analyze stability, $\mathcal H_\infty$- and $\mathcal H_2$-performance, as well as general quadratic performance of continuous-time uncertain systems is provided in the tutorial paper~\cite{Veenman2016}.
In addition to these performance measures, the continuous-time toolbox IQClab~\cite{Veenman2021} includes also the energy-to-peak gain analysis.
However, results on analyzing the peak-to-peak gain are missing in continuous time.
In discrete-time, the $\Hinf$-performance has been analyzed in~\cite{Hu2017} using IQCs with positive-negative multipliers and exponential stability has been analyzed in~\cite{Lessard2016,Schwenkel2020} using hard IQCs.
We extend these results to the larger class of finite horizon IQCs with terminal cost (see~\cite{Scherer2022a,Scherer2018} for the benefits of IQCs with terminal cost).
Further, our analysis results improve our previous results on robust peak-to-peak analysis using IQCs from~\cite{Schwenkel2023b}, and we obtain significantly tighter bounds than the IQC analysis of the energy-to-peak gain presented in~\cite{Jaoude2020}.

\emph{Related work on IQC synthesis.} The design of robust controllers is a challenging problem.
In fact, unlike the robust analysis problem, no convex optimization procedure is known that solves the general robust synthesis problem.
Nevertheless, there are useful approaches (without guarantees to find the global optimum) like $\mu$-synthesis~\cite{Balas1991} in the structured singular value framework or continuous-time IQC synthesis~\cite{Veenman2016a,Veenman2014,Veenman2014a,Wang2016}.
These methods are based on an iteration between a convex analysis step to optimize the involved multipliers and a convex synthesis step to optimize the controller parameters.
The big advantage of using IQCs over the structured singular value is that IQCs can describe more types of uncertainty such as time-varying parameters or sector-/slope-restricted nonlinearities.
Furthermore,~\cite{Apkarian2006} uses non-smooth optimization techniques to design a continuous-time controller via IQCs directly in the frequency domain. 
However, all these works can only optimize the $\Hinf$-performance but are not applicable to peak-to-peak or energy-to-peak gain minimization.
Going beyond $\mathcal H_\infty$-performance offers much more flexibility for a targeted controller design as in many applications the larger concern is to bound the peak rather than the energy of the output (see~\cite{Dahleh1994} for a detailed discussion).

\emph{Outline.} Section~\ref{sec:setup} formalizes the problem setup, introducing the necessary mathematical framework and definitions. In Section~\ref{sec:ana}, we provide LMI conditions to verify robust exponential stability and to compute robust upper bounds on the $\mathcal H_\infty$-norm, the energy-to-peak gain, and the peak-to-peak gain. 
Section~\ref{sec:synthesis} shows how to convexify the synthesis of controllers minimizing these performance bounds for a fixed IQC multiplier by using a suitable factorization of the IQC.
These two steps allow us to present an iterative algorithm in Section~\ref{sec:algo} which is investigated in numerical examples in Section~\ref{sec:exmp}.

\emph{Notation.}
To indicate that $A$ is defined to be equal to $B$ we write $A\triangleq B$. 
We denote the open unit disc in the complex plane by $\mathbb D\triangleq\{\lambda \in \mathbb{C} \mid |\lambda|<1\}$, its boundary by $\partial \mathbb D\triangleq\{\lambda \in \mathbb{C} \mid |\lambda|=1\}$, and its closure by $\overline{ \mathbb D} \triangleq \mathbb D \cup \partial \mathbb D$.
For a matrix $A\in \R^{n\times n}$, we denote the set of eigenvalues by $\lambda(A) \triangleq \{\lambda \in \mathbb C \mid \det(\lambda I -A)=0\}$.
For $x\in\R^n$, denote the infinity norm by $\|x\|_\infty \triangleq \max_{i} |x_i|$ and the Euclidean norm by $\|x\|_2 \triangleq \sqrt{x^\top x}$.
The frequency domain variable of the z-transformation is denoted by $\zs\in\mathbb{C}$.
For matrices $(A,B,C,D)$, we denote the corresponding transfer function by $\ssrep{A}{B}{C}{D}(\zs)\triangleq C(\zs I-A)^{-1}B+D$. 
For a transfer function define $G^*(\zs) \triangleq G (\zs^{-1})^\top$.
The set of proper, stable, rational transfer functions is denoted by $\RHinf = \left\{\left.\ssrep{A}{B}{C}{D}\,\right|\,\lambda(A) \subseteq \mathbb D\right\}$.
The set of symmetric matrices $A=A^\top \in \R^{n\times n}$ is denoted by $\S^n$.
If $A \in\S^{n}$ is a positive (semi-)definite matrix, we write $A\succ 0$ ($A\succeq 0$).
If $A \in\S^{n}$ is a negative (semi-)definite matrix, we write $A\prec 0$ ($A\preceq 0$).
For matrices $A\in \R^{n\times m}$ and $P \in\R^{i\times j}$, we denote $\diag(A,P)\triangleq \begin{pmatrix}A & \\  & P\end{pmatrix}\triangleq \begin{pmatrix}A & 0 \\ 0 & P\end{pmatrix}\in\R^{(n+i)\times (m+j)}$ and if $\P\in \S^{n}$, we denote $\symbscalar{PA}\triangleq A^\top P A$.
Furthermore, we use $\star$ (and $\bullet$) to denote symmetric (and irrelevant) entries in block matrices, e.g., $\begin{pmatrix} I & A \\ \star & I\end{pmatrix}=\begin{pmatrix} I & A \\ A^\top & I
\end{pmatrix}$, whereas $\begin{pmatrix} I & A \\ \bullet & I\end{pmatrix}=\begin{pmatrix} I & A \\ B & I \end{pmatrix}$ for some matrix $B$ with suitable dimensions. 
The set of sequences $x:\N\to\R^n$ is denoted by $\ell_{2\mathrm{e}}^n$.
For any operators $\Delta:\ell_{2\mathrm{e}}^{m_1}\to\ell_{2\mathrm{e}}^{n_1}$, $K:\ell_{2\mathrm{e}}^{m_2}\to\ell_{2\mathrm{e}}^{n_2}$, $G=\begin{pmatrix} G_{11} & G_{12}\\ G_{21} & G_{22} \end{pmatrix}$ with $G_{ij}:\ell_{2\mathrm{e}}^{n_i}\to\ell_{2\mathrm{e}}^{m_j}$ for $i,j\in\{1,2\}$, we define the upper linear fractional transformation by $\Delta \star G = G_{22} + G_{21} \Delta (I-G_{11} \Delta)^{-1} G_{12}$ and the lower linear fractional transformation by $G\star K = G_{11} + G_{12} K (I-G_{22} K)^{-1} G_{21}$ assuming the inverses $(I-G_{22} K)^{-1}$, $(I-G_{11} \Delta)^{-1} $ exist.

\section{Setup and Preliminaries}\label{sec:setup}
In this section, we describe the problem setting, we formally define our control goals in terms of stability and the performance measures $\Hinf$-norm, energy-to-peak, and peak-to-peak gain in Subsection~\ref{sec:goals}, we introduce a so-called loop transformation in Subsection~\ref{sec:loop}, and we recap the definition of IQCs in Subsection~\ref{sec:IQCs}.
We consider the problem of designing a controller $\K$ for the interconnection $\DG$ of a linear system $\G$ and a (possibly nonlinear) uncertainty $\Delta$ as depicted in Figure~\ref{fig:DGK}.
The uncertain system $\DG$ is given by\\\noindent
\begin{minipage}{0.4\linewidth}
    \begin{align*}
        G = \ssrepflex{c}{ccc}{\A\G & \B\G\p & \B\G\w & \B\G\u}
        {\C\G\q & \D\G\q\p & \D\G\q\w & \D\G\q\u \\
        \C\G\z & \D\G\z\p & \D\G\z\w & \D\G\z\u \\
        \C\G\y & \D\G\y\p & \D\G\y\w & 0},
    \end{align*}
\end{minipage}\hfill
\begin{minipage}{0.5\linewidth}
\begin{subequations}\label{eq:sys}
  \begin{flalign}
    \x_{k+1} &= \A\G \x_\k + \B\G\p   \p_\k + \B\G\w \w_\k + \B\G\u \u_\k \\
    \q_\k &= \C\G\q \x_\k  + \D\G\q\p \p_\k + \D\G\q\w \w_\k + \D\G\q\u \u_\k\\
    \z_\k &= \C\G\z \x_\k  + \D\G\z\p \p_\k + \D\G\z\w \w_\k +  \D\G\z\u \u_\k &\\
    \y_\k &= \C\G\y \x_\k  + \D\G\y\p \p_\k + \D\G\y\w \w_\k\\
    \p_k &= (\Delta(\q))_k
  \end{flalign}
\end{subequations}
\end{minipage}\\[0.2cm]
where $\k\in\N$ denotes the time index. 
The dimensions of the signals are $\x\in \ell_{2\mathrm{e}}^\nx$, $\p\in\ell_{2\mathrm{e}}^\np$, $\q \in \ell_{2\mathrm{e}}^\nq$, $\w\in\ell_{2\mathrm{e}}^\nw$, $\z\in\ell_{2\mathrm{e}}^\nz$, $\u\in\ell_{2\mathrm{e}}^\nu$, and $\y\in\ell_{2\mathrm{e}}^\ny$. 
The system matrices have suitable dimensions. 
We assume that the uncertainty $\Delta:\ell_{2\mathrm{e}}^\nq\to\ell_{2\mathrm{e}}^\np$ is a causal operator and belongs to a known set $\Delta\in\Deltaset$.
Furthermore, we assume that the interconnection~\eqref{eq:sys} is well-posed, i.e., that for all $\w\in\ell_{2\mathrm{e}}^\nw$, $\u\in\ell_{2,\mathrm{e}}^\nu$ and $\Delta\in\Deltaset$ there is a unique solution of~\eqref{eq:sys} that causally depends on $\w$ and $\u$.
The channel $\p\to\q$ is called the uncertainty channel, $\w\to\z$ the performance channel, and $\u \to\y$ the control channel.
We close control channel with a controller of the form\\\noindent
\begin{minipage}{0.4\linewidth}
    \begin{align*}
        K = \ssrepflex{c}{c}{\A\K & \BB\K }
        {\C\K & \DD \K },
    \end{align*}
\end{minipage}\hfill
\begin{minipage}{0.5\linewidth}
    \begin{subequations}\label{eq:K}
      \begin{flalign}
        \Kx_{k+1} &= \AK \Kx_k + \BK \y_k &\\
        \u_k &= \CK \Kx_k + \DK \y_k.
      \end{flalign}
    \end{subequations}
\end{minipage}\\[0.2cm]
In the synthesis problem, we optimize the controller parameters $\A\K$, $\BB\K$, $\CC\K$, and $\DD\K$ such that $\DGK$ is stable and minimizes some desired performance criterion.

\subsection{Performance criteria and stability}\label{sec:goals}
The goal of the controller synthesis is robust stability and robust performance in the sense that the closed loop is stable for all $\Delta\in\Deltaset$ and that we minimize the worst-case gain from $\w$ to $\z$ in some performance measure. 
In particular, this article considers three common performance criteria: the $\mathcal H_\infty$-norm, the energy-to-peak gain, and the peak-to-peak gain. 
The energy of a signal $\w\in\ell_{2\mathrm{e}}^\nw$ is its $\ell_2$-norm $\|\w\|_{2} \triangleq \sqrt{\sum_{k=0}^{\infty} \|\w_k\|_2^2}$ and the space of all bounded energy signals is $\ell_{2}^n\triangleq \{w\in\ell_{2\mathrm{e}}^n \mid \|w\|_{2} < \infty \}$.
More generally, for $\rho \in (0,1]$ we define the  $\ell_{2,\rho}$-norm $\|w\|_{2,\rho} \triangleq \sqrt{\sum_{k=0}^{\infty} \rho^{-2k}\|\w_k\|_2^2}$ and the  $\ell_{2,\rho}$-space $\ell_{2,\rho}^n\triangleq \{\w\in\ell_{2\mathrm{e}}^n \mid \|\w\|_{2,\rho} < \infty \}$, which contains for $\rho=1$ the standard $\ell_2$-norm and $\ell_2$-space.
The peak of a signal $w\in\ell_{2\mathrm{e}}^{\nw}$ is defined by $\|w\|_\peak = \sup_{k\geq 0} \|\w_k\|_2$. 
Now we are prepared to define the three performance measures of interest.
\begin{definition}[Energy and peak induced norms]\label{def:gain}
    For an operator $H:\ell_{2\mathrm{e}}^{\nw}\to\ell_{2\mathrm{e}}^\nz$ we define the $\mathcal H_\infty$-norm, the energy-to-peak (\ETP) gain, and the peak-to-peak (\PTP) gain by
    \begin{align}\label{eq:p2p}
        \|H\|_\ETE = \sup_{w\in\ell_{2}^\nw\setminus\{0\} } \frac{\| H(w)\|_2}{\|w\|_2}, \qquad
        \|H\|_\ETP = \sup_{w\in\ell_{2}^\nw\setminus\{0\} } \frac{\| H(w)\|_\peak}{\|w\|_2}, \qquad
        \|H\|_\PTP = \sup_{w\in\ell_{\infty}^\nw\setminus\{0\} } \frac{\| H(w)\|_\peak}{\|w\|_\peak}.
    \end{align}
\end{definition}
The gains in~\eqref{eq:p2p} view the system $H=\DGK$ as an operator mapping input signals $\w$ to output signals $\z$ with zero initial condition.
The $\mathcal H_\infty$-norm characterizes the energy-to-energy gain or $\ell_2$-gain.
For SISO systems $H$ it is known~\cite{Dahleh1994} that $\|H\|_\ETP$ is the $\mathcal H_2$-norm and that $\|H\|_\PTP$ is the $\ell_1$-norm of the impulse response of $H$. 
\begin{remark}[(Peak- and $\ell_\infty$-norm)]
    The $\ell_\infty$-norm, defined by $\|\w\|_\infty = \sup_{k\geq 0} \|\w_k\|_{\infty}$, is for scalar signals $\w\in\ell_{2\mathrm{e}}$ equal to the peak-norm $\|w\|_\peak = \|w\|_\infty$.
    Furthermore, for general $\w\in\ell_{2\mathrm{e}}^\nw$, we have $\|w\|_\infty \leq \|w\|_\peak \leq \sqrt{\nw} \|w\|_\infty$ (see, e.g.,~\cite{Rieber2008}).
    Define the $\ell_{\infty}$-to-$\ell_{\infty}$ gain $\|H\|_{\infty\shortrightarrow \infty}\triangleq \sup_{w\in\ell_{\infty}^\nw\setminus\{0\} } \frac{\| Hw\|_\infty}{\|w\|_\infty}$ and the energy-to-$\ell_\infty$ gain $\|H\|_{2\shortrightarrow\infty}=\sup_{w\in\ell_{2}^\nw\setminus\{0\} } \frac{\| Hw\|_\infty}{\|w\|_2}$.
    Then, for every operator $H:\ell_{2\mathrm{e}}^{\nw}\to\ell_{2\mathrm{e}}^\nz$, we can provide the following relations 
    \begin{align*}
        \frac{1}{\sqrt{\nz}}\|H \|_\PTP \leq \|H\|_{\infty\shortrightarrow\infty} \leq \sqrt{\nw} \|H \|_\PTP\qquad \text{and} \qquad 
        \frac{1}{\sqrt{\nz}}\|H \|_\ETP \leq\  \|H\|_{2\shortrightarrow\infty}\ \leq  \|H \|_\ETP.
    \end{align*}
    Due to this close connection, the peak-to-peak gain can also be used to minimize the $\ell_{\infty}$-to-$\ell_{\infty}$ gain, i.e., $\ell_1$ performance.
    Moreover, the peak norm is typically better suited for reachability analysis, as the obtained ellipsoidal reachable sets often have a smaller volume than the rectangular ones obtained via the $\ell_\infty$ norm.
\end{remark}

For general initial conditions with possibly $\x_0\neq 0$ and $\Kx_0\neq 0$, we are further interested in stability in the following sense. 
\begin{definition}[$\ell_{2,\rho}$-stability]\label{defn:stab}
    Let $\rho\in(0,1]$. The interconnection $\DGK$, described by~\eqref{eq:sys},~\eqref{eq:K}, is called $\ell_{2,\rho}$-stable, if there exists $c_0>0$ such that for all $w\in\ell_{2,\rho}^{\nw}$ and all $x_0\in\R^{\nx}, \kappa_0\in\R^\nK$ it holds that
    \begin{align}\label{eq:exp_stab}
        \left\|\begin{pmatrix} \x\\\Kx
        \end{pmatrix}\right\|_{2,\rho}^2 \leq c_0  \left(\left\|\begin{pmatrix} \x_0\\\Kx_0
        \end{pmatrix}\right\|^2_2  + \left\|\w\right\|^2_{2,\rho}\right).
    \end{align}
\end{definition}
\begin{remark}[(Exponential and input-to-state stability)]
    We remark that \eqref{eq:exp_stab} for all $w\in\ell_{2,\rho}^\nw$ implies
    \begin{align}\label{eq:exp_stab2}
    \sum_{k=0}^{t} \rho^{-2k}\left\|\begin{pmatrix} \x_k\\\Kx_k
		\end{pmatrix}\right\|^2_2 \leq c_0 \left(\left\|\begin{pmatrix} \x_0\\\Kx_0
		\end{pmatrix}\right\|^2_2 + \sum_{k=0}^{t-1} \rho^{-2k}\left\|\w_k\right\|^2_2\right)
    \end{align}
    for all $w\in\ell_{2\mathrm{e}}^\nw$ and all $t \in \N$, as any truncated signal is in $\ell_{2,\rho}$.
    Clearly, the opposite is also true as starting from~\eqref{eq:exp_stab2}, taking any $w\in\ell_{2,\rho}^\nw$ and letting $t\to \infty$ shows~\eqref{eq:exp_stab}. 
    Using this characterization, it can be shown by induction on $t$ that $\ell_{2,\rho_1}$-stability implies $\ell_{2,\rho_2}$-stability for all $1\geq\rho_2 \geq \rho_1$.
    Further, if $\rho<1$, then~\eqref{eq:exp_stab2} implies input-to-state stability 
    \begin{align}\label{eq:iss}
    \left\|\begin{pmatrix} \x_k\\\Kx_k
		\end{pmatrix}\right\|^2_2 \leq c_0 \rho^{2k} \left\|\begin{pmatrix} \x_0\\\Kx_0
		\end{pmatrix}\right\|^2_2 + c_0 \frac{\rho^2}{1-\rho^2} \left\|\w\right\|^2_\peak 
    \end{align}
    and for $\w=0$ we obtain $\rho$-exponential stability.
    Even for $\rho=1$, the $\ell_2$-stability~\eqref{eq:exp_stab2} is equivalent to input-to-state stability if $\Delta$ is static, i.e., if $\begin{pmatrix}x\\\Kx\end{pmatrix}$ is the full state of the uncertain system $\DGK$, as is shown in~\cite{Sontag1998}.
    For $\rho<1$ and static $\Delta$, we can interpret~\eqref{eq:exp_stab2} as exponential input-to-state stability (compare~\cite{Drummond2024}).
\end{remark}

\subsection{Loop transformation}\label{sec:loop}
To analyze the peak-to-peak gain and $\ell_{2,\rho}$-stability with contraction rate $\rho\in(0,1]$, we consider a transformed version of the control loop $\DGK$. 
This transformation is known from the robust exponential stability analysis using IQCs~\cite{Hu2016}.
For $a\in (0,\infty)$, define the linear transformation
\begin{align*}
  T_a : \ell_{2\mathrm{e}}^n \to \ell_{2\mathrm{e}}^n,\quad  (s_k)_{k\in\N}\, \mapsto\, (a^k s_k)_{k\in\N},
\end{align*}
which satisfies $T_a T_{a^{-1}} = I $. 
Now, define $\bar \x \triangleq T_{\rho^{-1}} \x$ and analogously define $\bar \p, \bar \w, \bar \u,\bar \q,\bar \z, \bar \y, \bar \Kx$.
These signals satisfy $ \|\bar \x\|_2 = \|x\|_{2,\rho}$ and thus $\bar \x \in \ell_{2} \Leftrightarrow x\in\ell_{2,\rho}$.
Moreover, define $\G_\rho \triangleq T_{\rho^{-1}} \circ \G \circ T_\rho$ and analogously define $\K_\rho$ as well as $\Delta_\rho$.
It is straightforward to verify that $\G_\rho$ has the state space representation
  \begin{align}\label{eq:sys_rho}
    \begin{pmatrix}\bar \x_{k+1} \\ \bar \q_\k \\ \bar\z_\k \\ \bar \y_\k \end{pmatrix} &= \begin{pmatrix}\rho^{-1}\A\G & \rho^{-1}\B\G\p & \rho^{-1}\B\G\w &\rho^{-1}\B\G\u \\[0.5mm]
      \C\G\q & \D\G\q\p &\D\G\q\w & \D\G\q\u \\[0.5mm]
      \C\G\z & \D\G\z\p &\D\G\z\w & \D\G\z\u \\[0.5mm]
      \C\G\y & \D\G\y\p &\D\G\y\w & 0 \end{pmatrix}\begin{pmatrix}\bar \x_\k \\ \bar \p_\k \\ \bar \w_\k \\ \bar \u_\k \end{pmatrix}
  \end{align}
and $\K_\rho$ has the state space representation
\begin{align}
  \begin{pmatrix} \bar \Kx_{\k+1} \\ \bar \u_\k \end{pmatrix} = \begin{pmatrix} \rho^{-1} \A\K & \rho^{-1} \BB\K \\ \CC\K & \DD\K \end{pmatrix}\begin{pmatrix} \bar \Kx_{\k} \\ \bar \y_\k \end{pmatrix}.
\end{align}
Therefore, we define $\A{\G_\rho} \triangleq \rho^{-1} \A\G$, $\A{\K_\rho} \triangleq \rho^{-1} \A\K$, $\B{\G_\rho}i \triangleq \rho^{-1} \B\G i$ for $i\in\{\p,\w,\u\}$, and $\BB{\K_\rho}\triangleq\rho^{-1}\BB\K$.
Furthermore, we have $\bar p = \Delta_\rho(\bar q)$ such that we obtain the transformed interconnection in Figure~\ref{fig:DGK_loop}.
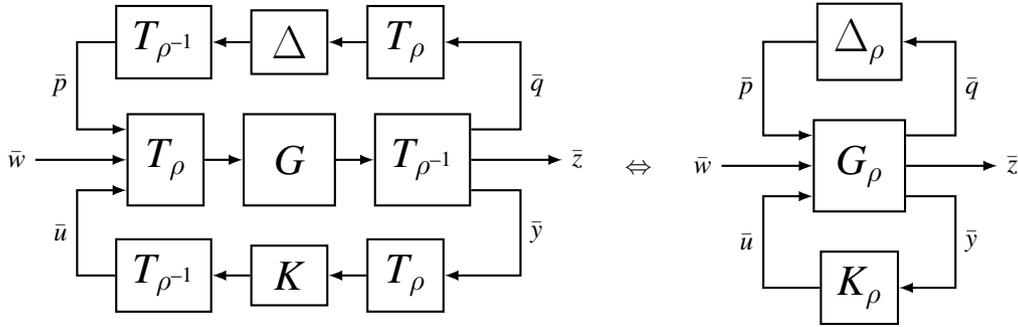
\begin{figure}[tb]\centering
	\begin{tikzpicture}[baseline=(G.base)]
            \node [block] (Delta) {\Large$\Delta$};
            \node [block, right=of Delta] (DeltaIn) {\Large$T_\rho$};
            \node [block, left=of Delta] (DeltaOut) {\Large$T_{\rho^{-1}}$};
            \draw[-latex,thick] (DeltaIn) -- (Delta);
            \draw[-latex,thick] (Delta) -- (DeltaOut);
            \node [block, below=of Delta, minimum height=1.2cm, minimum width=1.2cm] (G) {\Large $\G$};
            \node [block, right=of G, minimum height=1.2cm] (Gout) {\Large$T_{\rho^{-1}}$};
            \node [block, left=of G, minimum height=1.2cm] (Gin) {\Large$T_\rho$};
            \draw[-latex,thick] (Gin) -- (G);
            \draw[-latex,thick] (G) -- (Gout);
            \node [block, below = of G] (K) {\Large $K$};
            \node [block, right=of K] (Kin) {\Large$T_\rho$};  
            \node [block, left=of K] (Kout) {\Large$T_{\rho^{-1}}$};
            \draw[-latex,thick] (Kin) -- (K);
            \draw[-latex,thick] (K) -- (Kout);
            \draw[-latex, thick] (Gout.east)+(0,-0.4)  --  +(0.65,-0.4) |-node [right,pos=0.25] {$\bar\y$} (Kin);
            \draw[latex-, thick] (Gin.west) -- +(-1.2,0)node [left] {$\bar \w$};
            \draw[latex-, thick] (Gin.west)+(0,-0.4)  -- +(-0.65,-0.4)|- node [left,pos=0.25] { $\bar\u$} (Kout);
            \draw[-latex, thick] (Gout.east) +(0,0) -- ($(Gout.east)+(1.2,0)$) node [right] {$\bar\z$};
            \draw[latex-, thick] (DeltaIn.east) -| node[right, pos=0.75] {$\bar\q$} ($(Gout.east)+(0.65,0.4)$) -- ($(Gout.east)+(0,0.4)$);
            \draw[-latex, thick] (DeltaOut.west) -| node[left, pos=0.75] {$\bar\p$} ($(Gin.west)+(-0.65,0.4)$) -- ($(Gin.west)+(0,0.4)$);
        \end{tikzpicture}
        \quad $\Leftrightarrow$ \quad
        \begin{tikzpicture}[baseline=(G.base)]
            \node [block] (Delta) {\Large$\Delta_\rho$};
            \node [block, below=of Delta, minimum height=1.2cm, minimum width=1.2cm] (G) {\Large $\G_\rho$};
		\draw[latex-, thick] (G.west) -- +(-1.2,0)node [left] {$\bar \w$};
		\node [block, below = of G] (K) {\Large $K_\rho$};
		\draw[-latex, thick] (G.east)+(0,-0.4)  --  +(0.65,-0.4) |-node [right,pos=0.25] {$\bar\y$} (K);
		\draw[latex-, thick] (G.west)+(0,-0.4)  -- +(-0.65,-0.4)|- node [left,pos=0.25] { $\bar\u$} (K);
		\draw[-latex, thick] (G.east) +(0,0) -- ($(G.east)+(1.2,0)$) node [right] {$\bar\z$};
		\draw[latex-, thick] (Delta.east) -| node[right, pos=0.75] {$\bar\q$} ($(G.east)+(0.65,0.4)$) -- ($(G.east)+(0,0.4)$);
		\draw[-latex, thick] (Delta.west) -| node[left, pos=0.75] {$\bar\p$} ($(G.west)+(-0.65,0.4)$) -- ($(G.west)+(0,0.4)$);
	\end{tikzpicture}
	\caption{Interconnection after the loop transformation.}\label{fig:DGK_loop}
\end{figure}
Note that the special case $\rho=1$ corresponds to the original interconnection from Figure~\ref{fig:DGK} as $\G_1=\G$, $\K_1=\K$, and $\Delta_1=\Delta$.
We denote the feedback interconnection between $\G_\rho$ and $\K_\rho$ by $\GKrho\triangleq \G_\rho \star \K_\rho $ and compute a state space representation of $\GKrho$ as
\begin{align}
  &\blkmat{c:c}{\A\GKrho & \B\GKrho i \newblkdash \C\GK j & \D\GK j i} \triangleq \blkmat{cc:c}{\A\Grho + \B\Grho\u \DK \C\G\y & \B{\Grho}\u\CK & \B{\Grho} i + \B{\Grho}\u \DK \D\G\y i \\ 
    \BB{\Krho} \C\G\y & \A{\Krho} & \BB{\Krho} \D\G\y i \newblkdash 
    \C\G j +\D\G j\u \DK \C\G\y & \D\G j\u \CK & \D\G j i+\D\G j \u \DK \D\G\y i}\label{eq:allsys_end}
\end{align}
for $j\in\{\q,\z\}$ and $i\in\{\p,\w\}$.
Note that $\C\GK j$ and $ \D\GK j i$ are independent of $\rho$ thus we do not need the index $\rho$ here.
Furthermore, for $\rho=1$, we define $\GK \triangleq \GK_1$.
The utility of the loop transformation is that $\ell_{2,\rho}$-stability of $\Delta\star\GK$ is equivalent to $\ell_{2}$-stability of $\Delta_\rho\star\GKrho$.

\subsection{Discrete-time Integral Quadratic Constraints}\label{sec:IQCs}
To obtain robust analysis and synthesis methods, we need to characterize the uncertainty set $\Deltaset$ as tight and detailed as possible while still being numerically tractable.
A powerful framework that serves this purpose are IQCs, which enable the incorporation of knowledge about the structure and nature of the uncertainty.
For example, norm-bounded dynamic uncertainties, (time-varying) uncertain parameters, uncertain time delays, or sector- or slope-restricted static nonlinearities can be described using IQCs. 
The original IQC article~\cite{Megretski1997} and the tutorial~\cite{Veenman2016} contain large libraries of IQCs for different uncertainties in continuous-time.
These continuous-time results carry over to discrete time (see, e.g.,~\cite{Fetzer2017,Hu2017,Kao2012}).
Although IQCs were originally proposed as a frequency domain framework, the use of finite horizon IQCs with terminal cost provides a seamless link to the time-domain~\cite{Scherer2022a,Scherer2018}.
This is particularly useful, since the time-domain nature of the peak norm requires time-domain and state-space arguments.
\begin{definition}[Finite horizon IQC with terminal cost]\label{def:iqc}
  Let a mutliplier $\M\in\S^{\nfilty}$, a terminal cost matrix $\X\in\S^{\nfilt}$ and a filter $\filt \in \RHinf^{\nfilty \times (\nq+\np)}$ be given. 
  Moreover, let $\filt = \ssrep{A_\filt }{B_\filt}{C_\filt}{D_\filt}$ be a minimal state space representation with state $\filtx\in\ell_{2\mathrm{e}}^{\nfilt}$, i.e.,
  \begin{subequations}\label{eq:filt_dyn}
    \begin{align}
      \filtx_{\k+1} &= \A\filt \filtx_\k + \B\filt\q \bar \q_\k  + \B\filt\p \bar \p_\k\\
      \filty_\k &= C_\filt \filtx_\k + D_\filt^\q \bar \q_\k + D_\filt^\p \bar \p_\k
    \end{align}
  \end{subequations}
  with $\filtx_0=0$, $B_\filt\triangleq\begin{pmatrix}\B\filt\q & \B\filt\p
  \end{pmatrix}$, and $D_\filt\triangleq\begin{pmatrix}D_\filt^\q & D_\filt^\p
  \end{pmatrix}$.
  A causal operator $\Delta$ is said to satisfy the \textit{finite horizon IQC with terminal cost} defined by $(\A\filt,\BB\filt,\CC\filt,\DD\filt,\X,\M)$ if for all $\bar \q \in \ell_{2\mathrm{e}}^{\nq}$, $ \bar\p = \Delta(\bar\q)$, and $t\in\N$ it holds that
  \begin{align}\label{eq:iqc}
    \sum_{\k=0}^{t-1} \filty_\k^\top \M \filty_\k + \filtx_{t}^\top \X \filtx_{t} \geq 0.
  \end{align}
\end{definition}
Note that both $\M$ and $\X$ are typically indefinite.
As the terminal cost $\filtx_t^\top \X\filtx_t$ depends on the particular realization of $\filt$, we define an IQC by $(\A\filt,\BB\filt,\CC\filt,\DD\filt,\X,\M)$ rather than by $(\filt,\X,\M)$ or by $(\filt,\M)$ as is usually done if $\X=0$~\cite{Hu2017,Lessard2016}.

\begin{remark}[(Hard and soft IQCs)]
    The classic IQC literature (e.g.,~\cite{Megretski1997,Seiler2015}) distinguishes hard and soft IQCs, where a hard IQC is a finite horizon IQC with terminal cost $X=0$ and a soft IQC is an infinite horizon IQC, i.e., $X=0$ and~\eqref{eq:iqc} has to hold only for $t=\infty$ instead of for all $t\geq 1$.
    Any hard IQC implies a soft IQC whereas the opposite is generally false.
    In~\cite{Scherer2018} it has been shown under standard assumptions that for each soft IQC there exists a symmetric matrix $X$ such that the finite horizon IQC with terminal cost $X$ holds.
    Therefore, finite-horizon IQCs with a terminal cost encompass both hard and soft IQCs in a finite-horizon setting. 
\end{remark}

\begin{remark}[($\rho$-hard IQCs)]
    Note that $\rho$-hard IQCs as introduced by~\cite{Lessard2016} can be incorporated in Definition~\ref{def:iqc} by using the loop transformation, in particular, $\Delta$ satisfies a $\rho$-hard IQC defined by $(\filt,\M)$ if and only if $\Delta_\rho$ satisfies the finite horizon IQC defined by $(\rho^{-1}\A\filt,\rho^{-1}\BB\filt,\CC\filt,\DD\filt,0,\M)$ (compare \cite{Boczar2015}).	
\end{remark}

\section{Robustness analysis}\label{sec:ana}
In this section we derive LMI conditions to verify $\ell_{2,\rho}$-stability as well as bounds on the $\Hinf$-norm, the energy-to-peak gain, and the peak-to-peak gain based on the assumption that $\Delta_\rho$ satisfies a finite horizon IQC with a terminal cost.
For $\rho \in (0,1]$ we utilize the loop transformation and augment the system $\GKrho$ with the filter $\filt$ to obtain the following augmented system $\allrho$ with state $\chi_k=\begin{pmatrix}\filtx_k \\ \bar\x_k \\ \bar\Kx_k \end{pmatrix}\in\R^{n_\chi}$, $n_\chi=\nfilt+\nx+\nK$, state space representation
\begin{subequations}\label{eq:aug_sys}
    \begin{align}
        \chi_{\k+1} &= \A\allrho \chi_\k + \B\allrho\p \bar\p_\k + \B\allrho\w \bar \w_\k \\
        \filty_\k &= \C\all\filty \chi_\k + \D\all{\filty}\p \bar\p_\k + \D\all{\filty}\w \bar\w_\k\\
        \bar\z_\k &= \C\all\z \chi_\k  + \D\all\z\p \bar\p_\k + \D\all\z\w \bar\w_\k,
    \end{align}
\end{subequations}
initial condition $\chi_0=0$, and the matrices 
\begin{align}\label{eq:all}
  &\blkmat{c:c:c}{\A\allrho & \B\allrho\p & \B\allrho\w \newblkdash \C\all\filty & \D\all{\filty}\p & \D\all{\filty}\w \newblkdash \C\all\z & \D\all\z\p & \D\all\z\w }\triangleq\blkmat{cc:c:c}{
    \A\filt & \B\filt\q\C\GK\q & \B\filt\p+ \B\filt\q\D\GK\q\p &\B\filt\q\D\GK\q\w\\ 
    0 & \A\GKrho & \B\GKrho\p & \B\GKrho\w \newblkdash
    \C\filt{} & \D\filt{}\q\C\GK\q & \D\filt{}\p+ \D\filt{}\q\D\GK\q\p & \D\filt{}\q\D\GK\q\w \newblkdash
    0 & \C\GK\z & \D\GK\z\p & \D\GK\z\w }.
\end{align}
Note that $\C\all j$ and $\D\all ji$ for $j\in\{\filty,\z\}$ and $i\in \{\p,\w\}$ are independent of $\rho$ and thus we dropped the index for clarity and ease of notation.
Whenever we do not need the loop transformation, we work with $\rho=1$ and define $\all\triangleq \all_1$.
Further, let us define $\ubar \X \triangleq \diagmat{\X \\ & 0}\in\R^{n_\chi \times n_\chi}$, which yields $\chi_t^\top \ubar \X \chi_t = \filtx_t ^\top \X \filtx_t$.
The following LMIs imply $\ell_{2,\rho}$-stability.

\begin{theorem}[Stability]\label{thm:stability}
    Assume for some $\rho\in(0,1]$ and for all $\Delta\in\Deltaset$ that $\Delta_\rho$ satisfies the finite horizon IQC with terminal cost defined by $(A_\filt,B_\filt,C_\filt,D_\filt,\X,\M)$. 
    Further, assume that there exists $\P\in\S^{n_\chi}$ and $\mu\geq 0$ such that
    \begin{align}\label{eq:stab_LMI1}
        \symb {\diagmat{-\P\\&\P\\&&\M\\&&&- \mu I}\hspace{-0.1cm}
        \begin{pmatrix}
            I &  0 & 0\\
            \A\allrho & \B\allrho\p & \B\allrho\w \\[0.5mm]
            \C\all\filty & \D\all{\filty}\p & \D\all{\filty}\w \\
            0& 0&I
        \end{pmatrix}}&\prec 0\\
        \label{eq:stab_LMI2}
        P-  \ubar\X&\succ 0
    \end{align}
    hold. Then $\DGK$ is $\ell_{2,\rho}$-stable for all $\Delta \in \Deltaset$.
\end{theorem}
\begin{proof} 
    Multiply~\eqref{eq:stab_LMI1} from the left by $\begin{pmatrix}\chi_\k^\top & \bar\p_\k^\top & \bar\w_\k^\top\end{pmatrix}$ and from the right by its transpose.
    Then, using~\eqref{eq:aug_sys}, the left-hand side of~\eqref{eq:stab_LMI1} becomes
    \begin{align}\label{eq:ana_diss1}
        \delta_1(k)\triangleq\chi_{\k+1}^\top \P\chi_{\k+1}-\chi_\k^\top \P\chi_\k+\filty_\k^\top\M\filty_\k - \mu \|\bar \w_\k\|^2_2.
    \end{align}
    Since~\eqref{eq:stab_LMI1} holds strict there exists an $\varepsilon>0$ such that $ \delta_1(k)\leq -\varepsilon \|\chi_k\|^2_2 \leq -\varepsilon \left\|\begin{pmatrix}\bar \x_k \\\bar \Kx_k \end{pmatrix}\right\|^2_2 $.
    Next, we sum up $\delta_1(k)$ and use a telescopic sum argument
    \begin{align}\label{eq:tel_sum}
        \sum_{\k=0}^{\t-1}\big(\chi_{\k+1}^\top \P\chi_{\k+1}-\chi_\k^\top \P\chi_\k\big)\! = \chi_\t^\top \P\chi_\t - \chi_0^\top \P\chi_0
    \end{align} 
    to obtain
    \begin{align*}
        -\sum_{k=0}^{t-1} \varepsilon \left\|
        \begin{pmatrix}
            \bar \x_k \\\bar \Kx_k
        \end{pmatrix}\right\|^2_2\geq \sum_{\k=0}^{\t-1} \delta_1(\k) 
        &= \sum_{k=0}^{t-1} \left(\filty_k^\top \M\filty_k - \mu \|\bar \w_k\|^2_2 \right) +\chi_\t ^\top P \chi_\t - \chi_0^\top P \chi_0.
    \end{align*}
    Now we can use the IQC~\eqref{eq:iqc} to obtain
    \begin{align*}
        -\sum_{k=0}^{t-1} \varepsilon \left\|
        \begin{pmatrix}
            \bar \x_k \\\bar \Kx_k
        \end{pmatrix}\right\|^2_2 \geq \chi_\t ^\top \left(P-\begin{pmatrix} \X & 0 \\ 0 & 0 \end{pmatrix}\right) \chi_\t- \chi_0^\top P \chi_0 - \sum_{k=0}^{t-1} \mu \|\bar \w_k\|^2_2.
    \end{align*}
    Since $\filtx_0=0$ (compare Def.~\ref{def:iqc}), we know $\chi_0^\top P\chi_0 = \begin{pmatrix}
  	\bar \x_0 \\ \bar \Kx_0
  \end{pmatrix}^{\!\!\top}\! P_{22} \begin{pmatrix}
  \bar \x_0 \\ \bar\Kx_0 \end{pmatrix}$, where $P = \begin{pmatrix}
  P_{11} & P_{12} \\ P_{21} & P_{22}
  \end{pmatrix}$ with $P_{11} \in \S^{\nfilt}$.
  Further, due to~\eqref{eq:stab_LMI2}, we know that $P_{22}\succ 0$,  $\Pschur\triangleq P_{22} - P_{21}(P_{11}-\X)^{-1} P_{12}\succ 0$, and $P- \begin{pmatrix}
  \X & 0 \\ 0 & \Pschur
  \end{pmatrix}=\begin{pmatrix}
  \P_{11} -\X & \P_{12} \\ \P_{21} & P_{21}(P_{11}-\X)^{-1} P_{12}
  \end{pmatrix}\succeq 0$ as the Schur complement reveals. Hence 
  \begin{align}\label{eq:mid_stab_proof}
  	\sum_{k=0}^{t-1} \varepsilon \left\|\begin{pmatrix}
  		\bar \x_k \\\bar \Kx_k
  	\end{pmatrix}\right\|^2_2 + \begin{pmatrix}
  		\bar\x_t \\ \bar\Kx_t
  	\end{pmatrix}^{\!\!\top}\! \Pschur \begin{pmatrix}
  		\bar\x_t \\ \bar\Kx_t
  	\end{pmatrix} \leq \begin{pmatrix}
  	\bar \x_0 \\\bar \Kx_0
  \end{pmatrix}^{\!\!\top}\! P_{22} \begin{pmatrix}
  \bar \x_0 \\\bar \Kx_0
  \end{pmatrix} +\sum_{k=0}^{t-1} \mu \|\bar \w_k\|^2_2.
\end{align}
Define $c_0 \triangleq \frac{1}{\min\{\varepsilon,\min \lambda(\Pschur)\}} \max\{\mu, \max \lambda(P_{22})\}$, then
\begin{align*}
	\sum_{k=0}^{t} \left\|\begin{pmatrix}
		\bar \x_k \\\bar \Kx_k
	\end{pmatrix}\right\|^2 \leq c_0 \left(\left\|\begin{pmatrix}
		\bar \x_0 \\\bar \Kx_0
	\end{pmatrix} \right\|^2+\sum_{k=0}^{t-1} \|\bar \w_k\|^2\right).
\end{align*}
Finally, using the definition of $\bar x_k=\rho^{-k}x_k$, $\bar \Kx_k = \rho^{-k} \Kx_k$, $\bar\w_k=\rho^{-k} \w_k$, and using the fact that this inequality holds for all $\t\in\N$ proves $\ell_{2,\rho}$-stability~\eqref{eq:exp_stab}.
\end{proof}

The following theorems extend the stability analysis to a performance analysis by providing LMIs that verify bounds on the $\mathcal H_\infty$-norm, as well as the energy- and peak-to-peak gain.
\begin{theorem}[$\mathcal H_\infty$-norm]\label{thm:e2e}
    Let $\rho\in(0,1]$. 
    Assume for all $\Delta\in\Deltaset$ that $\Delta_\rho$ satisfies the finite horizon IQC with terminal cost defined by $(A_\filt,B_\filt,C_\filt,D_\filt,\X,\M)$. 
    Further, assume that there exist $\P\in\S^{n_\chi}$ and $\gamma\geq 0$ such that~\eqref{eq:stab_LMI2} and
    \begin{align}\label{eq:ETE_LMI}
        \symb {\diagmat{-\P\\&\P\\&&\M\\&&&\frac 1 \gamma  I\\&&&&- \gamma I}\hspace{-0.1cm}
        \begin{pmatrix}
            I &  0 & 0\\
            \A\allrho & \B\allrho\p & \B\allrho\w \\[0.5mm]
            \C\all\filty & \D\all{\filty}\p & \D\all{\filty}\w \\
            \C\all\z & \D\all{\z}\p & \D\all{\z}\w \\
            0& 0&I
        \end{pmatrix}}&\prec 0
    \end{align}
    hold. Then $\DGK$ is $\ell_{2,\rho}$-stable for all $\Delta \in \Deltaset$ and satisfies $\|\DGKrho\|_{\ETE} \leq \gamma $.
\end{theorem}
\begin{proof}
    Since~\eqref{eq:ETE_LMI} implies~\eqref{eq:stab_LMI1} with $\mu=\gamma$, we conclude using Theorem~\ref{thm:stability} that $\DGK$ is $\ell_2$-stable.
    Analogous to the proof of~\eqref{eq:stab_LMI1} to~\eqref{eq:mid_stab_proof} in Theorem~\ref{thm:stability}, inequality~\eqref{eq:ETE_LMI} implies with $\x_0=0$ and $\Kx_0=0$ that
    \begin{align*}
        0 \leq \sum_{k=0}^{t-1} \varepsilon \left\|\begin{pmatrix}
             \bar \x_k \\ \bar \Kx_k
        \end{pmatrix}\right\|^2_2 
        + \begin{pmatrix}
            \bar \x_t \\ \bar \Kx_t
        \end{pmatrix}^{\!\!\top}\! \Pschur \begin{pmatrix}
            \bar \x_t \\ \bar \Kx_t
        \end{pmatrix} 
        \leq \sum_{k=0}^{t-1} \left(\gamma\| \bar \w_k\|^2_2-\frac{1}{\gamma}\| \bar \z_k\|_2^2\right)
    \end{align*}
    holds.
    As this bound holds for all $\t\in\N$, we conclude $\|\bar \z\|_2\leq \gamma \|\bar \w\|_2$.
    Hence, $\|\DGKrho\|_\ETE \leq \gamma$ for all $\Delta\in\Deltaset$ by definition~\eqref{eq:p2p} and since $\bar \z = (\DGKrho) \bar w$.
\end{proof}
\begin{theorem}[Energy- and peak-to-peak gain]\label{thm:p2p}
  Assume for some $\rho\in(0,1]$ and for all $\Delta\in\Deltaset$ that $\Delta_\rho$ satisfies the finite horizon IQCs with terminal cost defined by $(A_\filt,B_\filt,C_\filt,D_\filt,\X_1,\M_1)$ and $(A_\filt,B_\filt,C_\filt,D_\filt,\X_2,\M_2)$. Further, assume that there exist $\P\in\S^{n_\chi}$, $\gamma\geq \mu\geq 0$ such that~\eqref{eq:stab_LMI1},~\eqref{eq:stab_LMI2} hold with $\M=\M_1+\M_2$, $\X=\X_1+\X_2$, and
  \begin{align}
     \label{eq:ana_LMI2}
    \symb{{\arraycolsep=1pt\begin{pmatrix}\ubar\X_1-\P\\
        &\ubar \X_2 \\
        &&\M_2\\
        &&&\frac {\alpha}{\gamma} I\\
        &&&&-\alpha(\gamma-\beta) I\end{pmatrix}}\hspace{-0.1cm}
    \begin{pmatrix}I&0&0\\[0.3mm]
      \A\allrho & \B\allrho\p & \B\allrho\w \\[0.8mm]
      \C\all\filty&\D\all\filty\p&\D\all\filty\w\\[0.8mm]
      \C\all\z & \D\all\z\p & \D\all\z\w \\[0.5mm]
      0&0&I\end{pmatrix}}  &\prec 0
  \end{align}
  where $\ubar \X_i \triangleq \diagmat{\X_i \\ & 0}$ for $i\in\{1,2\}$, as well as
  \begin{enumerate}[(i)]
      \item $\rho \in (0,1)$, $\beta = \mu$, and $\alpha=\frac{\rho^2}{1-\rho^2}$, then $\DGK$ is $\ell_{2,\rho}$-stable and $\|\DGK\|_\PTP \leq \gamma $;
      \item $\rho=1$, $\beta=0$, $\alpha=1$, and $\mu=\gamma$, then $\DGK$ is $\ell_{2}$-stable and $\|\DGK\|_{\ETP} \leq \gamma $.
  \end{enumerate}
\end{theorem}
\begin{proof}
    As we assumed~\eqref{eq:stab_LMI1} and~\eqref{eq:stab_LMI2} we can follow the proof of Theorem~\ref{thm:stability} to deduce that $\DGK$ is $\ell_{2,\rho}$-stable for all $\Delta \in \Deltaset$.
    Furthermore, as shown in~\eqref{eq:ana_diss1} in the proof of Theorem~\ref{thm:stability}, inequality~\eqref{eq:stab_LMI1} implies $\delta_1(k)\leq 0$.
    Similarly, we multiply~\eqref{eq:ana_LMI2} from the left by $\begin{pmatrix}\chi_\k^\top & \bar\p_\k^\top & \bar\w_\k^\top\end{pmatrix}$ and from the right by its transpose, and using~\eqref{eq:aug_sys} we obtain
  \begin{align} 
    \delta_2(k)&\triangleq-\chi_\k^\top \P\chi_\k+ \chi_\k^\top \ubar \X_1\chi_\k+\chi_{\k+1}^\top \ubar \X_2 \chi_{\k+1}+\filty_\k^\top\M_2\filty_\k +\frac {\alpha}{\gamma}\|\bar \z_\k\|^2_2  -\alpha(\gamma-\beta) \|\bar \w_\k\|^2_2\leq 0. \label{eq:ana_diss2}
  \end{align}
  Combining both inequalities $\delta_1(k)\leq 0$ and $\delta_2(k)\leq 0$ as follows
  \begin{align}\label{eq:comb_LMI1_2}
  \delta_2(\t) +\sum_{\k=0}^{\t-1} \delta_1(\k) \leq 0
  \end{align}
  and using the telescoping sum argument~\eqref{eq:tel_sum} with $\M=\M_1+\M_2$, and $\chi_0=0$ ($\filtx_0=0$ by Def.~\ref{def:iqc}, $x_0=0$, $\kappa_0=0$ by Def.~\ref{def:gain})
  leads to 
  \begin{align*}\nonumber
    0&\geq \frac{\alpha}{\gamma} \|\bar \z_\t\|^2_2-\alpha(\gamma-\beta) \|\bar \w_\t\|^2_2- \mu \sum_{\k=0}^{\t-1} \|\bar \w_\k\|^2_2  +\sum_{\k=0}^{\t-1} \filty_\k^\top \M_1\filty_\k+ \chi_{\t}^\top \ubar\X_1 \chi_{\t}  +\sum_{\k=0}^\t \filty_\k^\top \M_2\filty_\k+ \chi_{\t+1}^\top \ubar\X_2 \chi_{\t+1}.
  \end{align*}
  Note that $\chi_{k}^\top \ubar\X_i \chi_{k}=\filtx_{k}^\top \X_i \filtx_{k}$ such that we can use the IQC~\eqref{eq:iqc} for $\M_1,\X_1$, and horizon length $\t$ as well as for $\M_2,\X_2$, and horizon length $\t+1$ to obtain
  \begin{align}\label{eq:p2p_e2p_same}
    0&\geq \frac{\alpha}{\gamma} \|\bar \z_\t\|^2_2-\alpha(\gamma-\beta) \|\bar \w_\t\|^2_2- \mu \sum_{\k=0}^{\t-1} \|\bar \w_\k\|^2_2. 
  \end{align}
  Now, let us first consider case \emph{(i)}: We have $\beta=\mu$ and $\|\bar \z_t\|^2_2 = \rho^{-2t} \|\z_t\|^2_2$, $\|\bar \w_t\|^2_2 = \rho^{-2t} \|\w_t\|^2_2$. Since $\rho \in (0,1)$ and $\alpha =\frac{\rho^2}{1-\rho^2}$, we have the geometric sum $\sum_{\k=0}^{\t-1}\rho^{-2\k}=\rho^{-2(t-1)}\sum_{\k=0}^{\t-1}\rho^{2\k} \leq \frac{\rho^{-2(t-1)}}{1-\rhosq}= \alpha \rho^{-2t}$ and thus
  \begin{align}\label{eq:wpeak}
    \sum_{\k=0}^{\t-1} \|\bar \w_\k\|^2_2=\sum_{\k=0}^{\t-1}\rho^{-2\k} \|\w_\k\|^2_2 \leq  \|\w\|^2_\peak\sum_{\k=0}^{\t-1}\rho^{-2\k} \leq \alpha \rho^{-2t}\|\w\|^2_\peak.
  \end{align}
  Hence, it follows 
  \begin{align*}
      0\geq \alpha \rho^{-2t} \left(\frac 1 \gamma  \|\z_\t\|^2_2 -(\gamma-\mu) \|\w_\t\|^2_2- \mu \|\w\|^2_\peak\right).
  \end{align*} 
  Finally, dividing this inequality by $\alpha \rho^{-2t}>0$ and using $\gamma - \mu\geq 0$, we obtain
  \begin{align*}
    \frac 1 \gamma\|\z_\t\|^2 &\leq  (\gamma-\mu) \|\w_\t\|^2+ \mu \|\w\|^2_\peak\leq (\gamma-\mu) \|\w\|^2_\peak+\mu\|\w\|^2_\peak = \gamma \|\w\|^2_\peak .
  \end{align*} 
  As the above reasoning holds for all $\t\in\N$, we deduce $\|\z\|^2_\peak \leq  \gamma^2\|\w\|^2_\peak $, i.e., $\|\DGK\|_\PTP \leq \gamma$ for all $\Delta\in\Deltaset$.
  Next, we consider case \emph{(ii)}: 
  Due to $\rho=1$, we have $\bar \z = \z$ and $\bar \w = \w$. 
  We plug this, $\mu=\gamma$, $\alpha =1$, and $\beta=0$ into~\eqref{eq:p2p_e2p_same} and obtain
    \begin{align*}
        \frac 1 \gamma \|\z_t\|^2 \leq \gamma \sum_{k=0}^{t} \|\w_t\|^2 \leq \gamma \|w\|_2^2.
    \end{align*}
  As the above reasoning holds for all $\t\in\N$, we deduce $\|\z\|_\peak \leq  \gamma\|\w\|_2$ i.e., $\|\DGK\|_{\ETP} \leq \gamma$ for all $\Delta\in\Deltaset$.
\end{proof}
\begin{remark}[(Nominal analysis)]
    In the nominal case, i.e.,  $\np=\nq=\nfilty=\nfilt = 0$, Theorems~\ref{thm:e2e} and~\ref{thm:p2p} recover existing LMIs for the $\mathcal H_\infty$, energy-, and peak-to-peak gain analysis as a special case~\cite[Proposition~3.12,~3.15,~3.16]{Scherer2005}. 
\end{remark}

\section{Synthesis}\label{sec:synthesis}
In this section, we move from analysis to synthesis and derive a design procedure for controllers $\K$ that robustly stabilize $\DGK$ and minimize a desired performance criterion. 
The difficulty in the synthesis is that the matrix inequalities~\eqref{eq:stab_LMI1},~\eqref{eq:stab_LMI2},~\eqref{eq:ana_LMI2} we used for analysis are no longer linear when we take the controller $\K$ as a decision variable. 
In the continuous-time $\mathcal H_\infty$-synthesis via IQCs~\cite{Veenman2016a,Veenman2014,Veenman2014a,Wang2016} the controller synthesis can be made convex under certain assumptions on $\filt^*\M\filt$ and by fixing the multiplier $\M$.
In Subsection~\ref{sec:trafo}, we show that also the discrete-time synthesis problem can be transformed to a convex problem for all our performance measures if $\M$ is fixed.
Then, one can iterate between a synthesis step for a fixed $\M$ to optimize over $\K$ and an analysis step for a fixed $\K$ to optimize over $\M$.
The key idea to convexify the synthesis is to invert the channel $\p\to\filty$ of $\filt$, however, $\filt$ is typically neither invertible nor square. 
Hence, we introduce a suitable factorization of the IQC multiplier $\filt^*\M\filt=\hat \filt^*\hat \M\hat \filt$ with $\hat\filt$ invertible in Subsection~\ref{sec:fact}.

\subsection{On the factorization of $\filt ^*\M\filt$}\label{sec:fact}
To construct a suitable factorization of $\filt ^*\M\filt$, we make the following assumption.
\begin{assumption}[Multiplier class]\label{ass:fact}
  Let $\filt = \begin{pmatrix} \filt_1 & \filt_2 \end{pmatrix}$ with $\filt_1 \in \RHinf^{\nfilty\times \nq}$ and $\filt_2\in\RHinf^{\nfilty \times \np}$. The filters $\filt_1$, $\filt_2$ and the multiplier $\M$ satisfy the following frequency domain inequalities on $ \partial \mathbb{D}$
  \begin{align}\label{eq:fact_ass1}
    \filt_1^* \M \filt_1 &\succ 0 \\ \label{eq:fact_ass2}
    \filt_2^* \M \filt_2 - \filt_2^* \M \filt_1 \left(\filt_1^*\M \filt_1\right)^{-1}\filt_1^* \M \filt_2&\prec 0.
  \end{align} 
\end{assumption}
Assumption~\ref{ass:fact} is satisfied by a large class of IQCs, which, in particular, contains the class of strict positive-negative multipliers as used in~\cite{Hu2017,Seiler2015,Wang2016} (a strict positive-negative multiplier satisfies~\eqref{eq:fact_ass1} and $\filt^*_2 \M\filt_2 \prec 0$).
Under Assumption~\ref{ass:fact}, we can construct a new $\hat \M$ and $\hat \filt$ with desirable properties.
\begin{theorem}[Factorization]\label{thm:fact}
  Let Assumption~\ref{ass:fact} hold and let $\hat \M = \diag(I_\nq,-I_\np)$. Then there exist $\hat \filt\in\RHinf^{(\nq+\np)\times(\nq+\np)}$ with $\hat\filt^*\hat \M \hat \filt = \filt^*\M\filt$ and the structure
    \begin{align}\label{eq:abcd_hat}
        \hat\filt =\begin{pmatrix} \hat \filt_{11} & \hat \filt_{12}\\0&\hat\filt_{22} \end{pmatrix}= \ssrep{\hat A}{\hat B}{\hat C}{\hat D},\ \ \blkmat{c:c}{\hat A & \hat B\newblkdash \hat C & \hat D} \triangleq \blkmat{cc:cc}{ \hat A_1 & 0 &  \hat B_1 & 0 \\ 0 &  \hat A_2 & 0 & \hat B_2 \newblkdash {\hat C}_{11}& {\hat C}_{12}&{\hat D}_{11}& {\hat D}_{12} \\ 0 & {\hat C}_{22}&0 & {\hat D}_{22}}
    \end{align}
    where $\hat \filt_{22}= \ssrep{\hat A_2 }{\hat B_2}{\hat C_{22} }{\hat D_{22}} \in\RHinf^{\np\times\np}$ has a stable inverse, i.e., $(\hat \filt_{22})^{-1}\in \RHinf^{\np\times\np}$.
    Furthermore, there exists matrices $\hat C_\filt\in\R^{\nfilty \times \nfilthat}$ and $\Z\in\S^{\nfilthat}$ (called certificate of the factorization) such that ${\filt} = \ssrep{\hat A }{\hat B}{\hat C_\filt }{\DD\filt}$ with $\hat A$ and $\hat B$ from~\eqref{eq:abcd_hat}, $\DD\filt$ from~\eqref{eq:filt_dyn} and
    \begin{align}\label{eq:certificate}
        &\symb{\diagmat{-\Z \\&\Z\\&&\ \hat M}\!\begin{pmatrix}I&0\\  \hat A & \hat B \\ \hat C & \hat D \end{pmatrix}}= \symbscalar{\M \begin{pmatrix} \hat C_\filt & D_\filt \end{pmatrix}}.
    \end{align}
\end{theorem}
\begin{proof}
    This proof is constructive and thus serves at the same time as a procedure to compute the matrices in~\eqref{eq:abcd_hat} and~\eqref{eq:certificate}.
    The construction is based on the discrete-time algebraic Riccati equation (DARE) 
    \begin{align}\label{eq:dare}
        A^\top ZA-Z+Q-(A^\top ZB+S)(B^\top Z B+R)^{-1}(A^\top ZB+S)^\top = 0
    \end{align}
    for which the solution set is defined by $Z\in\dare(A,B,Q,R,S)\triangleq\{Z\in\S^{\nfilthat} \mid  \text{\eqref{eq:dare} and } B^\top Z B + R \text{ invertible}\}$.
    To construct the factorization, we perform the following steps.
  \begin{enumerate}[1.]
    \item\emph{Constructing $\hat \filt_{11}$.} Let $\filt_1=\ssrep{\A1}{\BB1}{\CC1}{\DD\filt^\q}$ be a minimal state space representation. Compute the unmixed\footnote{See~\cite{Clements2003} for a definition of unmixed solutions.} solution $Z_\mathrm{u}\in \dare(A_1, B_1, Q_1, R_1, S_1)$ with $\begin{pmatrix} Q_1 & S_1 \\ S^\top_1 & R_1\end{pmatrix} = \symbscalar{\M \begin{pmatrix} \C{1}{} & \D\filt\q{} \end{pmatrix}}$ 
    that satisfies $B^\top_1 Z_\mathrm{u} B_1 + R_1\succ 0$ and $\lambda(A_1-B_1\underline{\hat D}_{11}^{-1}\underline{\hat C}_{11})\subseteq  \{0\} \cup \{\lambda \in \mathbb C\mid |\lambda| >1 \}$ 
    with $\underline{\hat D}_{11}^\top \underline{\hat D}_{11} = B^\top_1 Z_\mathrm{u} B_1 + R_1$ and $\underline{\hat{C}}_{11}\triangleq\underline{\hat D}_{11}^{-\top} (A_1^\top Z_\mathrm{u} B_1 + S_1)^\top $.
    This unmixed solution exists and is unique due to~\eqref{eq:fact_ass1} as shown in Lemma~\ref{lem:dare} in the Appendix~\ref{appendix}. 
    It can be computed using the approach presented in~\cite{Ionescu1992} but by selecting the eigenvalues in $\{0\} \cup \{\lambda \in \mathbb C\mid |\lambda| >1 \}$ instead of the stable ones.
    Note that due to $Z_\mathrm{u}\in \dare(A_1, B_1, Q_1, R_1, S_1)$ and the definition of $\underline{\hat C}_{11}$ and $\underline{\hat D}_{11}$ we have 
    \begin{align}\label{eq:certificate11}
      &\symb{\diagmat{\Z_\mathrm{u} \\&-\Z_\mathrm{u}\\&& I}\!\begin{pmatrix}I&0\\  A_1 & B_1 \\ \underline{\hat C}_{11} & \underline{\hat D}_{11} \end{pmatrix}}= \begin{pmatrix}
        Q_1 & S_1 \\ S_1^\top & R_1
      \end{pmatrix}
      =\symbscalar{\M \begin{pmatrix} C_1 & \D\filt\q{} \end{pmatrix}}.
    \end{align}
    Define $ \underline{\hat \filt}_{11}\triangleq \ssrep{A_1}{B_1}{\underline{\hat C}_{11}}{\underline{\hat D}_{11}}$. 
    Then due to~\eqref{eq:certificate11} we have for all $\zs\in\mathbb C$ that
    \begin{align}
      \filt_1^*(\zs)\M\filt_1(\zs)=\symh{\symbscalar{\M \begin{pmatrix} C_{1}{} & \D\filt\q{} \end{pmatrix}} \begin{pmatrix} (\zs I-A_{1})^{-1}B_{1}{} \\ I \end{pmatrix}} &\refeq{\small \eqref{eq:certificate11}}{=}\symh{
        \symb{\diagmat{\Z_\mathrm{u} \\&-\Z_\mathrm{u}\\&&\ I}\begin{pmatrix}I&0\\  A_{1} &  B_{1}{} \\ \underline{\hat C}_{11} & \underline{\hat D}_{11} \end{pmatrix}} \begin{pmatrix} (\zs I- A_{1})^{-1} B_{1}{} \\ I \end{pmatrix}}\nonumber \\ \label{eq:fact_11}
      &\refeq{\cite[Lemma~2.3]{Stoorvogel1998}\hspace{0.8cm}}{=}\symh{\symbscalar{\begin{pmatrix} \underline{\hat C}_{11} & \underline{\hat D}_{11} \end{pmatrix}} \begin{pmatrix} (\zs I- A_{1})^{-1} B_{1}{} \\ I \end{pmatrix}}=\underline{\hat\filt}_{11}^*(\zs) \underline{\hat \filt}_{11}(\zs)
    \end{align}
    holds, i.e., $\filt_1^*\M\filt_1 = \hat {\underline{\filt}}_{11}^*\underline{\hat \filt}_{11}$.
    Let $n_0\geq 0$ be the algebraic multiplicity of the eigenvalue at $0$ of $A_1-B_1\underline{\hat D}_{11}^{-1}\underline{\hat C}_{11}$. 
    Define $\hat \filt_{11}(\zs) \triangleq \zs[-n_0] \underline{\hat \filt}_{11}(\zs)$, then $\hat\filt_{11}\in\RHinf^{\nq\times\nq}$ and 
    \begin{align}\label{eq:cancel_0_trick}
      \hat \filt_{11}^* \hat \filt_{11} = \underline{\hat\filt}_{11}^* \underline{\hat \filt}_{11}\refeq{\eqref{eq:fact_11}}{=} \filt_1^*\M\filt_1
    \end{align}
    due to $(\zs[-n_0])^*\zs[-n_0]=1$ for $\zs\in\mathbb C\setminus\{0\}$.
    Further, let $\begin{pmatrix}\hat \filt_{11}\\ \filt_{1}\end{pmatrix} = \ssrep{\hat A_1}{ \hat B_1}{ \hat C_{11} }{ \hat D_{11} \\ \hat C_{1} & \D\filt\q{}}$ be a minimal state space representation. 
    Then, $\hat A_1$ is Schur stable and $(\hat A_1,\hat B_1)$ is controllable.
    \item \emph{Constructing $\hat \filt_{12}$.} We have $\hat \filt_{11}^{-*} \filt_1^* = (\filt_1 \hat \filt_{11}^{-1})^* = (\zs[n_0] \filt_1\underline{\hat \filt}^{-1}_{11})^*$ and further (by removing uncontrollable modes) 
    \begin{align*}
      \filt_1\underline{\hat \filt}^{-1}_{11} = \ssrepflex{cc}{c}{A_1 & -B_1 \underline{\hat D}_{11}^{-1} \underline{\hat C}_{11} & B_1 \underline{\hat D}_{11}^{-1} \\ 0& A_1-B_1 \underline{\hat D}_{11}^{-1} \underline{\hat C}_{11} & B_1 \underline{\hat D}_{11}^{-1}  }{ C_1 & -\D\filt\q{}\underline{\hat D}_{11}^{-1}\underline{\hat C}_{11} & \D\filt\q{}\underline{\hat D}_{11}^{-1} } = \ssrep{A_1-B_1 \underline{\hat D}_{11}^{-1} \underline{\hat C}_{11} }{B_1 \underline{\hat D}_{11}^{-1} }{ C_1-\D\filt\q{}\underline{\hat D}_{11}^{-1}\underline{\hat C}_{11} }{ \D\filt\q{}\underline{\hat D}_{11}^{-1} }.
    \end{align*}
    From 1. we know $\lambda \left( A_1-B_1\underline{\hat D}_{11}^{-1}\underline{\hat C}_{11}\right) \subseteq \{0\} \cup \{\lambda \in \mathbb C\mid |\lambda| >1 \}$. Thus, $\filt_1\underline{\hat \filt}^{-1}_{11}$ has $n_0$ poles at $0$ and all other poles outside $\overline {\mathbb D}$.
    Hence, all poles of $\filt_1 \hat \filt_{11}^{-1} =\zs[n_0] \filt_1\underline{\hat \filt}^{-1}_{11}$ lie outside $\overline {\mathbb D}$ as we canceled the $n_0$ poles at $0$.
    Therefore, $\hat \filt_{11}^{-*} \filt_1^*$ is causal and has all poles inside $\mathbb D$, i.e., $\hat \filt_{11}^{-*} \filt_1^*\in\RHinf^{\nq\times\nfilty}$.
    Therefore, also 
    \begin{align}\label{eq:filt12}
        \hat \filt_{12} \triangleq\hat \filt_{11}^{-*} \filt_1^* M\filt_2\in\RHinf^{\nq\times\np},
    \end{align}
    which implies that there is a minimal state space representation of $\begin{pmatrix}\hat \filt_{12} \\ \filt_2 \end{pmatrix}=\ssrepflex{c}{c}{\hat A_2 & \hat B_2}{ \hat C_{12} & \hat D_{12} \\ \hat C_2 & \D\filt\p{}}$ with $\hat A_2$ Schur stable and $(\hat A_2,\hat B_2)$ controllable.
     
    \item \emph{Constructing $\hat \filt_{22}$.} Note that
    \begin{align*}
      &\symh{\begin{pmatrix}I & 0 \\ 0 & -\M \end{pmatrix} \begin{pmatrix} \hat \filt_{12} \\ \filt_2 \end{pmatrix}}
      =\hat \filt_{12}^* \hat \filt_{12}- \filt_{2}^*\M \filt_{2} \overset{\small \eqref{eq:cancel_0_trick}, \eqref{eq:filt12}}{=} \filt_2^* \M \filt_1 \left(\filt_1^*\M \filt_1\right)^{-1}\filt_1^* \M \filt_2-\filt_2^* \M \filt_2  \succ 0
    \end{align*}
    where the last frequency domain inequality holds on $\partial \mathbb{D}$ due to \eqref{eq:fact_ass2}.
    Hence, we can apply Lemma~\ref{lem:dare} in the Appendix~\ref{appendix} which provides a unique stabilizing solution $Z_{\mathrm{s}} \in \dare(\hat A_2, \hat B_2, Q_2,R_2,S_2)$ with $\begin{pmatrix} Q_2 & S_2 \\ S^\top_2 & R_2\end{pmatrix} = \symb{\begin{pmatrix} I & 0 \\ 0 & -\M \end{pmatrix} \begin{pmatrix} \hat C_{12} & \hat D_{12} \\  \hat C_2 & \D\filt\p{} \end{pmatrix}}$. 
    Since $\hat B^\top_2 Z_{\mathrm{s}} \hat B_2 + R_2\succ 0$, there exists $\hat D_{22}$ with $\hat D_{22}^\top \hat D_{22} = \hat B^\top_2 Z_{\mathrm{s}} \hat B_2 + R_2$.
    Further, define $\hat{C}_{22}=\hat D_{22}^{-\top} (\hat A_2^\top Z_{\mathrm{s}} \hat B_2 + S_2)^\top $.
    Then, we know that $\lambda \left( \hat A_2-\hat B_2{\hat D}_{22}^{-1}{\hat C}_{22}\right) \subseteq \mathbb{D}$ as $Z_{\mathrm{s}}$ is the stabilizing solution.
    Thus, $\hat \filt_{22} = \ssrep{\hat A_2}{\hat B_2}{\hat C_{22}}{\hat D_{22}}$ has a stable inverse $\hat \filt_{22}^{-1}=\ssrep{\hat A_2-\hat B_2{\hat D}_{22}^{-1}{\hat C}_{22} }{\hat B_2 \hat D_{22}^{-1}}{-\hat D_{22}^{-1} \hat C_{22} }{\hat D_{22}^{-1} } \in \RHinf^{\np\times \np}$.
    Furthermore, the definition of $\hat D_{22}$ and $\hat C_{22}$ as well as the fact that $Z_{\mathrm{s}} \in \dare(\hat A_2, \hat B_2, Q_2,R_2,S_2)$ yield
    \begin{align*}
      &\symb{\diagmat{\Z_{\mathrm{s}} \\&-\Z_{\mathrm{s}}\\&& I}\begin{pmatrix}I&0\\  \hat A_2 & \hat B_2 \\ \hat C_{22} & \hat D_{22} \end{pmatrix}}= \symb{\begin{pmatrix} I & 0 \\ 0 & -\M \end{pmatrix} \begin{pmatrix} \hat C_{12} & \hat D_{12} \\ \hat C_2 & \D\filt\p{} \end{pmatrix}}.
    \end{align*}
    Analogous to the arguments in step 1., this guarantees $\hat \filt_{22}^*\hat \filt_{22} = \hat \filt_{12}^*\hat \filt_{12}-\filt_2^*\M\filt_2$.
    Thus, we have shown that $\hat \filt^*\hat \M \hat \filt$ is indeed a factorization, i.e.,
    \begin{align}
      \hat \filt^*\hat \M \hat \filt &= \begin{pmatrix} \hat \filt_{11}^* \hat \filt_{11} &  \hat \filt_{11}^*  \hat \filt_{12} \\ \hat \filt_{12}^* \hat \filt_{11} &  \hat \filt_{12}^* \hat \filt_{12}- \hat \filt_{22}^* \hat \filt_{22} \end{pmatrix} \overset{\small \eqref{eq:cancel_0_trick}, \eqref{eq:filt12}}{=} \begin{pmatrix} \filt_1^* \M \filt_1 &  \filt_1^* \M \filt_2 \\ \filt_2^* \M \filt_1 &  \filt_2^* \M \filt_2 \end{pmatrix}=\filt^* \M \filt. \label{eq:fact}
    \end{align}
    \item \emph{Computing the certificate $\Z$.} Define $\hat C_\filt=\begin{pmatrix} \hat C_1 & \hat C_2 \end{pmatrix}$. 
    Then $\Z$ is the solution of the Stein equation 
    \begin{align}\label{eq:stein}
        \hat A^\top \Z\hat A- \Z +\hat C^\top\hat \M\hat C - \hat C_\filt ^\top \M\hat C_\filt = 0.
    \end{align}
    To show that this is the case, we write~\eqref{eq:fact} compactly as
    \begin{align}\label{eq:fde_for_certificate}
      0=\begin{pmatrix} \hat \filt \\ \filt \end{pmatrix}^* \begin{pmatrix} -\hat M & 0 \\ 0 & M \end{pmatrix}\begin{pmatrix} \hat \filt \\  \filt \end{pmatrix}\ \text{with}\ 
      \begin{pmatrix} \hat \filt \\ \filt \end{pmatrix} = \ssrepflex{c}{c}{ \hat A &  \hat B }{ \hat C &\hat D \\ \hat C_\filt & D_\filt}.
    \end{align}
    Since we constructed $\hat A_1$, $\hat A_2$ to be Schur stable and $(\hat A_1,\hat B_1)$, $(\hat A_2,\hat B_2)$ to be controllable, we know due to~\eqref{eq:abcd_hat} that also $\hat A$ is Schur stable and $(\hat A,\hat B)$ is controllable.
    Hence, we can apply Lemma~\ref{lem:dare_zero} in the Appendix~\ref{appendix} to~\eqref{eq:fde_for_certificate}, where  the matrices $(A_i,B_i,C_i,D_i,M)$ in Lemma~\ref{lem:dare_zero} for both $i\in\{1,2\}$ are $\left(\hat A, \hat B, \begin{pmatrix} \hat C \\ \hat C_\filt\end{pmatrix}, \begin{pmatrix} \hat D \\ \hat D_\filt\end{pmatrix}, \begin{pmatrix} -\hat M & 0 \\ 0 & M    \end{pmatrix}\right)$.
    Then, the Lemma yields a solution $Z$ that satisfies 
    \begin{align*}
      &\symb{\diagmat{-\Z \\&\Z}\begin{pmatrix}I&0\\  \hat A & \hat B \end{pmatrix}}= \symb{\begin{pmatrix} -\hat M & 0 \\ 0 & \M \end{pmatrix} \begin{pmatrix} {\hat C} &{\hat D} \\ \hat C_\filt & D_\filt \end{pmatrix}}.
    \end{align*}
    The upper left block is the Stein equation~\eqref{eq:stein} and uniquely determines $\Z$.
    When bringing $\symbscalar{\hat M \begin{pmatrix} \hat C & \hat D\end{pmatrix}}$ to the other side, we obtain~\eqref{eq:certificate} and hence, $\Z$ indeed certifies the factorization.\\[-1.35cm]
  \end{enumerate}
\end{proof}
\begin{remark}
    The factorization of Theorem~\ref{thm:fact} is known from continuous time~\cite{Veenman2014}, with the technical difference that Theorem~\ref{thm:fact} requires $\hat\filt$ and $\hat \filt_{22}^{-1}$ to be stable, whereas~\cite{Veenman2014} requires $\hat \filt_{11}$ and $\hat \filt^{-1}$ to be stable.
    However, constructing the factorization in discrete time is not a trivial extension of the continuous-time case, as the complex conjugate $\filt^*$ of a causal discrete-time system $\filt$ with a pole at $0$ is acausal. 
    In contrast, in continuous time the complex conjugate $\filt^*$ is always casual. 
    We fixed this problem by introducing additional poles at $0$ in $\hat\filt_{11}$ such that $\hat\filt_{11}^{-*}$ is causal.
    This change, however, requires a new way to find the certificate $\Z$ for which we proposed Lemma~\ref{lem:dare_zero} in the Appendix~\ref{appendix}.
    In continuous time it is even possible to avoid the use of factorizations as shown in the synthesis procedures~\cite{Veenman2016a,Veenman2014a} by working directly with a state space representation of $\filt^*\M\filt$ instead of one for $\filt$, which can have numerical benefits.
    In discrete time, however, if $\filt$ has poles at the origin, $\filt^*\M\filt$ is improper and cannot be represented by a state space model but only by a descriptor state model.
    To avoid the difficulties of descriptor state systems, we work with the factorization approach from~\cite{Veenman2014,Wang2016}. 
\end{remark}

This new factorization can be equivalently used to describe the uncertainty $\Delta$ if the terminal cost is suitably modified as the following theorem shows.

\begin{theorem}\label{thm:term_cost}
    Consider $\hat \filt$, $\hat \M$, $\Z$ from Theorem~\ref{thm:fact}.
    Then, there exists $\hat V\in\R^{\nfilt\times\nfilthat}$ with full row rank that satisfies $\hat V \hat A = A_\filt \hat V$, $ \hat V\hat B = B_\filt $ as well as $\hat C_\filt =C_\filt \hat V$.
    Further, a causal operator $\Delta:\ell_{2\mathrm{e}}^\nq\to\ell_{2\mathrm{e}}^\np$ satisfies the finite horizon IQC with terminal cost defined by $(A_\filt,B_\filt,C_\filt,D_\filt,\X,\M)$ if and only if it satisfies the one defined by $(\hat A,\hat B,\hat C,\hat D,\hat X,\hat \M)$ with $\hat X=\hat V^\top \X\hat V + Z$.
\end{theorem}
\begin{proof}
    First, note that $\ssrep{\A\filt}{\BB\filt}{\CC\filt}{\DD\filt}=\filt$ is a minimal state space representation whereas the realization $\ssrep{\hat A}{\hat B}{\hat C_\filt }{D_\filt}=\filt$ contains unobservable modes.
    Hence, there exists a transformation matrix $\hat V\in\R^{\nfilt\times \nfilthat}$ with full row rank, that satisfies 
    $\hat V \hat A = A_\filt \hat V$, $ \hat V\hat B = B_\filt $ as well as $\hat C_\filt =C_\filt \hat V$. 
    Moreover, if $\hat \filtx$ denotes the state corresponding to the state space realization with $(\hat A,\hat B)$, then $\filtx = \hat V\hat \filtx$.
    Let $\hat s = \hat \filt \begin{pmatrix}\q\\\p\end{pmatrix}$ and use~\eqref{eq:certificate} to observe that
  \begin{align*}
    \sum_{\k=0}^{\t-1} \filty_\k^\top \M\filty_\k
    &= \sum_{\k=0}^{\t-1} \symb{\symbscalar{\M \blkmat{c:c}{ \hat C_\filt & D_\filt }} \begin{pmatrix}\hat \filtx_\k \newblkdash[1ex] \q_\k \\ \p_\k \end{pmatrix}} \refeq{\eqref{eq:certificate}}{=}\ \sum_{\k=0}^{\t-1} \symb{\symb{\diagmat{- \Z \\&\Z\\&&\ \hat M}\blkmat{c:c}{I&0\\  \hat A & \hat B \\ \hat C & \hat D }} \begin{pmatrix}\hat \filtx_\k \newblkdash[1ex] \q_\k \\ \p_\k \end{pmatrix}} \\
    &= \sum_{\k=0}^{\t-1} \symb{\diagmat{- \Z \\&\Z\\&&\ \hat M}\begin{pmatrix}\hat \filtx_\k \\ \hat \filtx_{k+1} \\ \hat \filty_\k \end{pmatrix}} = \sum_{\k=0}^{\t-1} \hat \filty_k^\top \hat \M \hat \filty_k + \hat \filtx_t^\top \Z \hat \filtx_t - \underbrace{\hat \filtx_0^\top \Z \hat \filtx_0}_{=0}.
  \end{align*}
  Since $\filtx = \hat V \hat \filtx$, we have $\filtx^\top_t \X\filtx_t = \hat \filtx_t^\top \hat V^\top \X\hat V\hat \filtx_t$ and thus 
  \begin{align*}
      0\refeq{\eqref{eq:iqc}}\leq \sum_{\k=0}^{\t-1} \filty_\k^\top \M\filty_\k + \filtx^\top_t \X\filtx_t = \sum_{\k=0}^{\t-1} \hat \filty_\k^\top \hat \M\hat \filty_\k + \hat \filtx^\top_t (\hat V^\top \X \hat V + Z)\hat \filtx_t
  \end{align*}
  which proves the statement.
\end{proof}

Theorem~\ref{thm:fact} allows us to use the analysis LMI in Theorem~\ref{thm:stability} and~\ref{thm:e2e} with the new factorization $\hat\M$, $\hat \filt$, and $\hat \X$ instead of the old one $\M$, $\filt$, and $\X$.
For Theorem~\ref{thm:p2p}, however, we need to restrict $(\M_1,\M_2,\X_1,\X_2)$ to 
\begin{align}\label{eq:M1M2_restriction}
    \M_1=(1-\sigma)\M,\quad \M_2 = \sigma \M,\quad\X_1 = (1-\sigma)\X, \quad\X_2=\sigma \X
\end{align}
with some fixed $\sigma\in[0,1]$ in order to have a common $\hat \filt$ and $\hat \M_1 = (1-\sigma)\hat \M$, $\hat X_1 = (1-\sigma)\hat X$, $\hat \M_2 = \sigma \hat \M$, and $\hat \X_2 = \sigma \hat \X$. 
To indicate the use of $\hat \filt$ instead of $\filt$, define $\allhatrho$ analogous to $\allrho$ in~\eqref{eq:aug_sys} and~\eqref{eq:all} but with $\hat A,\hat B,\hat C,\hat D$ instead of $\A\filt,\BB\filt,\CC\filt,\DD\filt$.
Note that $\D\all\z \p=\D\allhat\z\p$ and $\D\all\z \w=\D\allhat\z\w$. 
Hence, for ease of notation and to indicate this equivalence, we always use $\D\all\z \p$ and $\D\all\z \w$.
We need the restriction~\eqref{eq:M1M2_restriction} to ensure that we can work with \emph{one} factorization and obtain only \emph{one} system $\allhatrho$.
Next, we show that Theorems~\ref{thm:e2e} and~\ref{thm:p2p} provide the same upper bound $\gamma$ no matter which of the two different factorizations of the IQC is used.
\begin{theorem}\label{thm:warm_start}
    Consider $\hat \filt$, $\hat \M$, $\Z$ from Theorem~\ref{thm:fact}.
    Then, the following statements hold.
    \begin{enumerate}[(i)]
        \item There exists $\P$ such that the LMIs~\eqref{eq:stab_LMI1},~\eqref{eq:stab_LMI2},~\eqref{eq:ana_LMI2}, and~\eqref{eq:M1M2_restriction} hold, if and only if there exists $\hat \P$ such that these LMIs hold with $(\allrho, \M, \X, \P)$ substituted by $(\allhatrho, \hat \M, \hat \X, \hat \P)$. 
        \item There exists $\P$ such that the LMIs~\eqref{eq:stab_LMI2} and~\eqref{eq:ETE_LMI} hold, if and only if there exists $\hat \P$ such that these LMIs hold with $(\allrho, \M, \X, \P)$ substituted by $(\allhatrho, \hat \M, \hat \X, \hat \P)$.
    \end{enumerate}
\end{theorem}
\begin{proof}
    We show this result by explicitly transforming the LMIs. 
    For the sake of conciseness, we exploit that the LMIs~\eqref{eq:stab_LMI1},~\eqref{eq:ETE_LMI}, and~\eqref{eq:ana_LMI2} restricted to~\eqref{eq:M1M2_restriction} are all of the form (stated in~$\hat\ $ variables)
    \begin{align}\label{eq:standard_LMI_form}
        \symb{{\arraycolsep=1pt\begin{pmatrix}(1-c_1)\hat{\ubar\X}-\hat \P\\
            &c_1 \hat{R} \\
            &&c_1 \hat \M\\
            &&&c_2 I\\
            &&&&c_3 I\end{pmatrix}}\hspace{-0.1cm}
        \begin{pmatrix}I&0&0\\[0.3mm]
            \A{\allhatrho} & \B{\allhatrho}\p & \B{\allhatrho}\w \\[0.8mm]
            \C{\allhat}{\hat\filty} & \D{\allhat}{{\hat\filty}}\p & \D{\allhat}{{\hat\filty}}\w \\[0.5mm]
            \C{\allhat}\z & \D\all\z\p & \D\all\z\w \\[0.5mm]
          0&0&I\end{pmatrix}} \prec 0
    \end{align}
    where for~\eqref{eq:stab_LMI1} we have $c_1=1$, $\hat R=\hat \P$, $c_2 = 0$, and $c_3 = - \mu$; for~\eqref{eq:ETE_LMI} we have $c_1=1$, $\hat R=\hat \P$, $c_2 = \frac{1}{\gamma}$, $c_3 = -\gamma$, and $\rho=1$; as well as for~\eqref{eq:ana_LMI2} we have $c_1=\sigma$, $\hat R=\hat{\ubar\X}$, $c_2 = \frac{\alpha}{\gamma}$, and $c_3 = -\alpha(\gamma-\beta)$.
    First, we transform the~$\hat \ $ version of the LMIs to the ones without~$\hat\ $ in three steps.
    \emph{Step 1: Transform $\hat \M$ and $\hat\filty$ to $\M$ and $\filty$.}
    Define $\tilde \P\triangleq \hat\P- \ubar Z$,  $\tilde \X \triangleq \hat \X-Z=\hat V^\top \X \hat V$, and $\tilde R \triangleq \hat R-\ubar \Z$, where the bar under a variable name denotes $\ubar Z=\begin{pmatrix} Z & 0 \\ 0 & 0 \end{pmatrix}$.
    Then, plugging~\eqref{eq:certificate} into~\eqref{eq:standard_LMI_form} yields
    \begin{align}\label{eq:ana_LMI_tilde}
        \symb{
            {\arraycolsep=1pt\begin{pmatrix}
                (1-c_1)\tilde{\ubar\X}-\tilde \P\\
                &c_1 \tilde{R} \\
                &&c_1 \M\\
                &&&c_2 I\\
                &&&&c_3 I
            \end{pmatrix}}\hspace{-0.1cm}
            \begin{pmatrix}I&0&0\\[0.3mm]
                \A{\allhatrho} & \B{\allhatrho}\p & \B{\allhatrho}\w \\[0.8mm]
                \C{\allhat}{\filty} & \D{\all}{{\filty}}\p & \D{\all}{{\filty}}\w \\[0.5mm]
                \C{\allhat}\z & \D\all\z\p & \D\all\z\w \\[0.5mm]
                0&0&I
            \end{pmatrix}
        }  \prec 0 
    \end{align}
    with $C_{\allhat}^\filty\triangleq \begin{pmatrix}\hat C_\filt & D_\filt^\q\C\GK\q\end{pmatrix}$.
    \emph{Step 2: Display unobservable modes.} Let $\hat V^\perp$ be an orthonormal basis of the kernel of $\hat V$,  i.e., $\hat V \hat V^\perp=0$ and $\hat V^{\perp\top}\hat V^\perp=I_{\nfilthat-\nfilt}$.
    Furthermore, let $\hat V ^\dagger \triangleq  \hat V^\top(\hat V \hat V^\top )^{-1}$ be the right inverse of $\hat V$.
    This implies $\hat V^{\perp\top}\hat V^\dagger = 0 $ and hence $\begin{pmatrix}\hat V^\perp & \hat V^{\dagger} \end{pmatrix}^{-1} = \begin{pmatrix}\hat V ^{\perp\top} \\ \hat V \end{pmatrix}$.
    Next, we do some preparatory computations to be able to multiply the matrix inequality~\eqref{eq:ana_LMI_tilde} from right by $T_1\triangleq\begin{pmatrix} \hat V^\perp & \hat V^\dagger & 0 \\ 0 & 0 & I\end{pmatrix}$ and from left by $T_1^\top$.
    Remember that $\hat V \hat A = \A\filt \hat V$ and $\hat C_\filt =C_\filt \hat V $, which yields $\hat V \hat A  \hat V^\dagger = A_\filt$, $\hat V \hat A  \hat V^\perp = 0$, $\hat C_\filt \hat V^\dagger = \CC\filt $, as well as $\hat C_\filt \hat V^\perp= 0$.
    Together with $ \hat V\hat B = B_\filt $ this yields
    \begin{align*}
    T_1^{-1}
        \blkmat{cc}{
            \hat A & \hat B \\
            \hat C_\filt  & \DD\filt 
        } T_1 = 
        \blkmat{c:c}{
            \hat V^{\perp\top} & 0 \\ 
            \hat V& 0   \newblkdash 
            0 & I 
        }
        \blkmat{c:c}{
            \hat A & \hat B \newblkdash
            \hat C_\filt  & \DD\filt 
        } 
        \blkmat{cc:c}{ 
            \hat V^\perp & \hat V^\dagger& 0 \newblkdash
            0 & 0 & I
        } &= \blkmat{cc:c}{
            A_\mathrm{u} &\bullet &\bullet \\
            0 & \A\filt & \BB\filt \newblkdash 
            0 & \CC\filt & \DD\filt
        } 
    \end{align*}
    where $A_\mathrm{u} \triangleq \hat V^{\perp\top} \hat A \hat V^\perp$ contains the unobservable modes.
    With the same transformation we can display these unobservable modes in $\allhatrho$ (remember that $\allhatrho$ is defined as $\allrho$ in \eqref{eq:all} but with $\hat A,\hat B,\hat C,\hat D$ instead of $\A\filt,\BB\filt,\CC\filt,\DD\filt$) 
    \begin{align*}
        \blkmat{cc}{
            T_1^{-1} & 0 \\ 0 & I 
        }
        \blkmat{c:c:c}{
            \A{\allhatrho} & \B{\allhatrho}\p & \B{\allhatrho}\w \newblkdash
            \C{\allhat}\filty &\D{ \all}\filty\p&\D{ \all}\filty\w\newblkdash
            \C{\allhat}\z & \D{\all}\z\p & \D{\all}\z\w 
        } 
        \blkmat{cc}{ 
            T_1 & 0 \\ 0 & I
        } &= \blkmat{cc:c:c}{
            A_\mathrm{u} &\bullet &\bullet&\bullet \\
            0 & \A\allrho & \B\allrho\p & \B\allrho \w \newblkdash 
            0 & \C\all\filty & \D\all\filty\p & \D\all\filty\w \newblkdash 
            0 & \C\all\z &\D{\all}\z\p & \D{\all}\z\w 
        }.
    \end{align*}
    Define $\bar P \triangleq T_1^{\top} \tilde P T_1$, $\bar R \triangleq T_1^\top \tilde R T_1$, and $\bar {\ubar \X}\triangleq T_1^{\top} \ubar{\tilde \X }T_1= \begin{pmatrix}
        0 & 0 \\ 0 & \ubar \X
    \end{pmatrix}$.
    Hence, multiplying the matrix inequality~\eqref{eq:ana_LMI_tilde} from right by $T_1$ and from left by $T_1^\top$ leads to
    \begin{align}\label{eq:ana_LMI_bar}
        \symb{
            {\arraycolsep=1pt\begin{pmatrix}
                (1-c_1)\bar{\ubar\X}-\bar \P\\
                &c_1 \bar{R} \\
                &&c_1 \M\\
                &&&c_2 I\\
                &&&&c_3 I
            \end{pmatrix}}\hspace{-0.1cm}
            \begin{pmatrix}
                I&0&0&0\\
                0&I &  0 & 0\newblkdash
                A_\mathrm{u} &\bullet &\bullet&\bullet \\
                0&\A{\allrho} & \B{\allrho}\p & \B{\allrho}\w \newblkdash
                0&C_\all^\filty & \D{\all}{\filty}\p & \D{\all}{\filty}\w \\
                0&\C{\all}\z & \D\all\z\p & \D\all\z\w \\
                0&0& 0&I
            \end{pmatrix}
        }  \prec 0.
    \end{align}
    \emph{Step 3: Remove unobservable modes.}
    Let  $\bar P = \begin{pmatrix}\bar P_{11} & \bar P_{12} \\ \bar P_{12}^\top & \bar P_{22} \end{pmatrix}$ be partitioned such that $\bar P_{11} \in \mathbb{S}^{\nfilthat-\nfilt}$.
    In \emph{(i)}, the LMI~\eqref{eq:stab_LMI1} holds, and in \emph{(ii)}, the LMI~\eqref{eq:ETE_LMI} holds.
    Recall that in both of these LMIs, we have $c_1=1$ and $\hat R=\hat \P$.
    Hence, the upper left block of~\eqref{eq:ana_LMI_bar} is $A_u^\top \bar P_{11} A_u -\bar P_{11}\prec 0$, which implies $\bar P_{11} \succ 0$ since $A_u$ is stable.
    Thus, for both \emph{(i)} and \emph{(ii)}, we can conclude $\bar P_{11} \succ 0$ and define 
    \begin{align}\label{eq:Pdef}
        \P\triangleq \bar \P_{22}-\bar \P_{12}^\top\bar P_{11}^{-1}\bar P_{12}\in\mathbb{S}^{n_\chi}
    \end{align} and $T_2 = 
            \begin{pmatrix} 
                I & -\bar P_{11}^{-1}\bar P_{12} \\ 
                0 & I 
            \end{pmatrix}$  yielding $T_2^\top
            \begin{pmatrix}
                \bar P_{11} & \bar P_{12} \\ 
                \bar P_{12}^\top & \bar P_{22} 
            \end{pmatrix} T_2
        = \diag(\bar P_{11} , P )  $, $T_2^\top  \bar{\ubar X} T_2 = \bar{\ubar X}=\diag(0,\ubar X)$, and thus also $T_2^\top \bar{R} T_2 = \diag(\bar R_{11}, R)$ as $\bar R\in\{\bar P, \bar{\ubar{X}}\}$. 
        Further, $
        T_2^{-1} \begin{pmatrix} A_\mathrm{u} & \bullet \\ 0 & \A\allrho \end{pmatrix}T_2 = \begin{pmatrix} A_\mathrm{u} & \bullet \\ 0 & \A\allrho \end{pmatrix}$, as well as $T_2^{-1} \begin{pmatrix} \bullet & \bullet \\ \B\allrho\p & \B\allrho\w \end{pmatrix} =  \begin{pmatrix} \bullet & \bullet \\ \B\allrho\p & \B\allrho\w \end{pmatrix}$.
    Therefore, when we multiply~\eqref{eq:ana_LMI_bar} from right with $\diag(T_2,I)$ and from left with its transpose, we obtain
    \begin{align}\label{eq:ana_LMI_trafo}
        \symb {
            \diagmat{-\bar \P_{11}\\&(1-c_1)\ubar X-\P \\ &&c_1 \bar R_{11} \\&&&c_1 R \\&&&&c_1 \M\\&&&&&c_2 I\\&&&&&&c_3 I}
            \begin{pmatrix}
                I&0&0&0\\
                0&I &  0 & 0\\
                A_\mathrm{u} &\bullet &\bullet&\bullet \\
                0&\A{\allrho} & \B{\allrho}\p & \B{\allrho}\w \\
                0&C_\all^\filty & \D{\all}{\filty}\p & \D{\all}{\filty}\w \\
                0&\C{\all}\z & \D\all\z\p & \D\all\z\w \\
                0&0& 0&I
            \end{pmatrix}
        }&\prec 0.
    \end{align}
    Since $\bar R_{11} \in\{0,\bar P_{11}\}$, we know $\bar R_{11}\succeq 0$ and thus, $\symbscalar{\bar R_{11} \begin{pmatrix}\bullet\end{pmatrix}}\succeq 0$.
    Hence, multiplying~\eqref{eq:ana_LMI_trafo} from left by $\begin{pmatrix}0 & I_{\nfilt+\nx+\nK+\np+\nw}\end{pmatrix}$ and from right by its transpose yields~\eqref{eq:standard_LMI_form} without $ \hat\ $.
    Finally, we note that $\P-\ubar \X = \begin{pmatrix}
        0&I
    \end{pmatrix}T_2^\top T_1^\top (\hat \P-\hat {\ubar X}) T_1 T_2\begin{pmatrix}
        0&I
    \end{pmatrix}^\top \succ 0 $.
    This concludes the proof of the if-part of \emph{(i)} and \emph{(ii)}.
    
    To show the only-if-part, we start with $\P$ solving $\P-\ubar X\succ 0$ as well as~\eqref{eq:stab_LMI1},~\eqref{eq:ana_LMI2},~\eqref{eq:M1M2_restriction} or~\eqref{eq:ETE_LMI}.
    Further, choose $W\succ 0$ such that $W-A_\mathrm{u}^\top WA_\mathrm{u} \prec 0$, which exists since $A_\mathrm{u}$ is stable. 
    Then, for some $\varepsilon>0$ let $\bar P_{11} = \varepsilon W$, $\bar P_{12}=0$, $\bar P_{22} = P$, and $\bar \X = \diag(0,\X)$.
    Since~\eqref{eq:stab_LMI1} and~\eqref{eq:ETE_LMI} hold strict, we can find $\varepsilon$ small enough such that~\eqref{eq:ana_LMI_bar} holds with $c_1=1$, $\bar R=\bar\P$ and $c_2=0$, $c_3=-\mu$, or respectively $c_2=\frac{1}{\gamma}$, $c_3=-\gamma$.
    Furthermore, LMI~\eqref{eq:ana_LMI2} restricted to~\eqref{eq:M1M2_restriction} and $\bar P_{11}=\varepsilon W\succ 0$ directly imply~\eqref{eq:ana_LMI_bar} with $\bar R_{11}=0$, $c_1=\sigma$, $R=\ubar X$, and $c_2 = \frac{\alpha}{\gamma}$, $c_3 = -\alpha(\gamma-\beta)$.
    Since we obtained~\eqref{eq:ana_LMI_bar} through equivalence transformations from~\eqref{eq:standard_LMI_form}, we can undo these transformations and obtain 
    \begin{align}\label{eq:Phat}
        \hat \P = \ubar Z +T_1^{-\top} \begin{pmatrix}
            \varepsilon W & 0 \\ 0 & P
        \end{pmatrix} T_1^{-1} .
    \end{align}
    Finally, note that $\hat \P - \hat {\ubar X} = T_1^{-\top}\diag(\varepsilon W, P-\ubar X )T_1^{-1} \succ 0 $, which concludes the only-if-part of \emph{(i)} and \emph{(ii)}.
\end{proof}
This constructive proof also enables computing suitable warm-starts for the alternating synthesis and analysis steps of our algorithm in Section~\ref{sec:algo}, which is crucial to show monotonicity of the iteration.

\subsection{Convexifying transformation of the decision variables}\label{sec:trafo}
We start this subsection by defining the transformed system and the convex synthesis LMIs before showing in Theorem~\ref{thm:syn} how to reconstruct the controller $\K$ in the original variables from a solution of these LMIs.
In particular, we take the open loop of $\allhatrho$ and invert $\hat \filt_{22}$ to obtain $\genplant$ given by 
\begin{align}\label{eq:Theta}
    &\blkmat{c:c:c:c}{
        \A{\genplantrho}& \B{\genplantrho}{\hat\filty_2} & \B{\genplantrho}\w & \B{\genplantrho}\u \newblkdash 
        \C\genplant{\hat\filty_1}& \D\genplant{\hat\filty_1}{\hat\filty_2} & \D\genplant{\hat\filty_1}\w & \D\genplant{\hat\filty_1}\u\newblkdash
        \C\genplant{\z}& \D\genplant\z{\hat\filty_2} & \D\genplant\z\w &\D\genplant\z\u\newblkdash
        \C\genplant{\y}& \D\genplant\y{\hat\filty_2} & \D\genplant\y\w &0
    }
    \triangleq\blkmat{ccc:c:c:c}{
        \hat A_1 & -\hat B_1 \D\G\q\p \hat D_{22}^{-1} \hat C_{22} & \hat B_1 \C\G\q & \hat B_1 \D\G\q\p\hat D_{22}^{-1} &\hat B_1 \D\G\q\w & \hat B_1 \D\G\q\u \\ 
        0 & \hat A_2-\hat B_2\hat D_{22}^{-1} \hat C_{22} & 0 & \hat B_2 \hat D_{22}^{-1} & 0 & 0 \\
        0 &-\B\Grho\p \hat D_{22}^{-1} \hat C_{22} & \A\Grho & \B\Grho\p \hat D_{22}^{-1} & \B\Grho\w & \B\Grho\u \newblkdash
        \hat C_{11} & \hat C_{12}-(\hat D_{12}+ \hat D_{11} \D\G\q\p)\hat D_{22}^{-1} \hat C_{22} & \hat D_{11}\C\G\q & (\hat D_{12}+ \hat D_{11} \D\G\q\p)\hat D_{22}^{-1} & \hat D_{11}\D\G\q\w & \hat D_{11} \D\G\q\u \newblkdash
        0 & -\D\G\z\p\hat D_{22}^{-1} \hat C_{22} & \C\G\z & \D\G\z\p\hat D_{22}^{-1} & \D\G\z\w & \D\G\z\u\newblkdash
        0 & -\D\G\y\p\hat D_{22}^{-1} \hat C_{22} & \C\G\y & \D\G\y\p\hat D_{22}^{-1} & \D\G\y\w & 0
    }.
\end{align}
The controller to be synthesized is of order $\nK = \nfilthat + \nx$.
To arrive at convex synthesis inequalities we transform the decision variables $(\P,\A\K,\BB\K,\CC\K,\DD\K)$ as shown in \cite[Section~4.2.2]{Scherer2005}.
The new variables $\Ysyn,\Xsyn\in\S^\nK$, $\Ksyn\in\R^{\nK\times \nK},\Lsyn\in\R^{\nK\times \ny},\Msyn\in\R^{\nu\times \nK},\Nsyn\in\R^{\nu\times \ny}$ are defined by
\begin{align}\label{eq:P11U}
    \begin{pmatrix}\Xsyn & \Usyn \end{pmatrix} &\triangleq \begin{pmatrix}I_{\nK}& 0_\nK \end{pmatrix} \hat\P, \\\label{eq:P22V}
    \begin{pmatrix}\Ysyn & \Vsyn \end{pmatrix} &\triangleq \begin{pmatrix}I_{\nK}& 0_\nK \end{pmatrix} \hat\P^{-1},\\\label{eq:KLMN}
    \begin{pmatrix} \Ksyn  & \Lsyn \\ \Msyn & \Nsyn\end{pmatrix} 
    &\triangleq \begin{pmatrix}\Usyn & \Xsyn\B{\genplantrho}\u \\ 0&I \end{pmatrix}
    \begin{pmatrix} \A{\K_\rho} & \B{\K_\rho}{} \\ \C{\K}{} & \D{\K}{}{} \end{pmatrix}
    \begin{pmatrix}\Vsyn^\top & 0 \\ \C\genplant\y \Ysyn & I\end{pmatrix} 
    + \begin{pmatrix} \Xsyn \A{\genplantrho} \Ysyn & 0 \\ 0 & 0 \end{pmatrix}
\end{align}
and appear linearly in the transformed closed loop
\begin{align}
    \blkmat{c:c}{\Asyn & \Bsyn i \newblkdash \Csyn {j\vphantom{j^j}} & \Dsyn ji} 
    &\triangleq \blkmat{cc:c}{
        \A{\genplantrho}\Ysyn+\B{\genplantrho}\u\Msyn & \A{\genplantrho}+\B{\genplantrho}\u\Nsyn\C\genplant\y & \B{\genplantrho} i+\B{\genplantrho}\u\Nsyn\D\genplant\y i\\ 
        \Ksyn & \Xsyn\A{\genplantrho}+\Lsyn\C\genplant\y& \Xsyn\B{\genplantrho} i +\Lsyn \D\genplant\y i \newblkdash
        \C\genplant{j}\Ysyn+\D\genplant{j}\u\Msyn & \C\genplant{j} +\D\genplant{j}\u\Nsyn\C\genplant\y & \D\genplant{j}i +\D\genplant{j}\u\Nsyn\D\genplant\y i
    }\label{eq:Asyn}
\end{align}
for $i\in\{\hat \filty_2,\w\}$ and $j\in\{\hat\filty_1,\z\}$.
The variables $\mu$ and $\gamma$ remain untransformed.
Further, define $\Psyn \triangleq \begin{pmatrix}\Ysyn & I \\ I& \Xsyn \end{pmatrix}$ and $\blkmat{c:c:c:c}{\mathcal E_1 & \mathcal E_2 & \mathcal E_3 & \mathcal E_4} \triangleq \begin{pmatrix} I_{\nfilthat} & 0
\end{pmatrix}\blkmat{c:c:c:c}{\Psyn & \Asyn & \Bsyn{\hat \filty_2} & \Bsyn\w}$ in which the decision variables also appear linearly.
Given a solution $\P^\mathrm{old},\mu^\mathrm{old},\gamma^\mathrm{old}$ of the analysis LMIs for some controller $\K^\mathrm{old}$, we can initialize the synthesis variables with a feasible warm start $(\Psyn^\mathrm{old},\Ksyn^\mathrm{old},\Lsyn^\mathrm{old},\Msyn^\mathrm{old},\Nsyn^\mathrm{old},\mu^\mathrm{old},\gamma^\mathrm{old})$ computed by~\eqref{eq:Phat},~\eqref{eq:P11U},~\eqref{eq:P22V},~\eqref{eq:KLMN}.
Based on this warm start, we define $\mathcal E_1^\mathrm{old}$, $\mathcal E_2^\mathrm{old}$, $\mathcal E_3^\mathrm{old}$, $\mathcal E_4^\mathrm{old}$ analogous to $\mathcal E_1$, $\mathcal E_2$, $\mathcal E_3$, $\mathcal E_4$, which are used for a convex relaxation that is needed if $\hat \X$ is not positive semi-definite.
In particular, decompose 
\begin{align}\label{eq:hatX_decomposition}
    \hat \X = L_1^\top L_1 - L_2^\top L_2
\end{align}
with $\Zdtwo = L_2^\top L_2$.
Now, we are ready to state the following synthesis LMIs for $\sigma\in(0,1)$ (the case $\sigma\in\{0,1\}$ is discussed in Remark~\ref{rem:rho_opt})
\begin{align}
    \begin{pmatrix}
        \Psyn+ \symb{\begin{pmatrix}
         -  \Zdtwo &  \Zdtwo\\ \Zdtwo & 0
    \end{pmatrix} \begin{pmatrix}
        \mathcal E_{1}^\mathrm{old}\\
        \mathcal E_{1} 
    \end{pmatrix}}   & \star \\ L_1\mathcal E_{1} & I
    \end{pmatrix}&\succ 0 \label{eq:syn_LMI_PX_pos_def}\\
        \begin{pmatrix}- \Psyn & 0 & 0 & \star & \star & \star  \\ 
        0 & -I & 0 & \star & \star & \star \\ 0&0&-\gamma I&\star & \star & \star \\ 
        \Asyn & \Bsyn{\hat \filty_2} &\Bsyn\w & -\Psyn & 0 & 0\\ \Csyn\filtyone & \Dsyn\filtyone{\hat \filty_2} & \Dsyn\filtyone\w & 0 & -I & 0\\ 
        \Csyn\z & \Dsyn\z{\hat \filty_2} & \Dsyn\z\w & 0 & 0 & -\gamma I
        \end{pmatrix}&\prec 0\label{eq:syn_LMI_e2e}\\\label{eq:syn_LMI_stability}
    \begin{pmatrix}- \Psyn & 0 & 0 & \star & \star \\ 0 & -I & 0 & \star & \star \\ 0&0&-\mu I&\star & \star \\ \Asyn & \Bsyn{\hat \filty_2} &\Bsyn\w & -\Psyn & 0 \\ \Csyn\filtyone & \Dsyn\filtyone{\hat \filty_2} & \Dsyn\filtyone\w & 0 & -I \end{pmatrix}&\prec 0 \\ 
            \begin{pmatrix}
                -\Psyn &0 & 0 & \star & \star & \star &\star \\
                0&-\sigma I &0 & 0 & \star & \star & \star\\
                0&0&-\alpha(\gamma-\beta) I &0 & \star & \star & \star \\             
                L_1\mathcal{E}_{1} & 0 & 0 & -\frac{1}{1-\sigma} I &0 & 0 & 0\\
                L_1\mathcal{E}_{2} & L_1\mathcal E_{3} & L_1\mathcal{E}_{4} &0& -\frac 1 \sigma I &0 & 0\\
                \Csyn{\hat\filty_1} & \Dsyn{\hat\filty_1}{\hat \filty_2} & \Dsyn{\hat\filty_1}\w &0&0& -\frac 1 \sigma I & 0 \\
                \Csyn\z & \Dsyn\z{\hat \filty_2} & \Dsyn\z\w &0&0&0& -\frac{\gamma}{\alpha} I
            \end{pmatrix}+
        \symb{
        \begin{pmatrix}
            (1-\sigma) \Zdtwo &0& \star & 0 \\ 
            0 & \sigma  \Zdtwo &0& \star \\  
            -(1-\sigma) \Zdtwo & 0 & 0 & 0\\
            0 & -\sigma  \Zdtwo &0 & 0 \\
        \end{pmatrix}
            \begin{pmatrix}
                \mathcal{E}_{1}^\mathrm{old} & 0 & 0 &0 \\
                \mathcal{E}_{2}^\mathrm{old} & \mathcal E_{3}^\mathrm{old} & \mathcal{E}_{4}^\mathrm{old} &0\\
                \mathcal{E}_{1} & 0 & 0 & 0\\
                \mathcal{E}_{2} & \mathcal E_{3} & \mathcal{E}_{4} & 0 \\
            \end{pmatrix}
        } &\prec 0.
     \label{eq:syn_LMI_e2p_p2p}
\end{align}
Note that these inequalities are linear in the decision variables $\Xsyn$, $\Ysyn$, $\Ksyn$, $\Lsyn$, $\Msyn$, $\Nsyn$, $\gamma$, and $\mu$.
The following theorem shows that the synthesis LMIs indeed provide a solution for the analysis LMIs and vice versa.
\begin{theorem}[Synthesis]\label{thm:syn}
    Let $\sigma\in(0,1)$.
    
    \begin{enumerate}[(i)]
        \item Suppose the $\mathcal H_\infty$-analysis LMIs~\eqref{eq:stab_LMI2},~\eqref{eq:ETE_LMI} [or energy- and peak-to-peak analysis LMIs~\eqref{eq:stab_LMI1},~\eqref{eq:stab_LMI2},~\eqref{eq:ana_LMI2} with~\eqref{eq:M1M2_restriction}] hold with $(\K,\P)=(\K^\mathrm{old},\P^\mathrm{old})$.
        Then, the $\mathcal H_\infty$-synthesis LMIs~\eqref{eq:syn_LMI_PX_pos_def},~\eqref{eq:syn_LMI_e2e} [or energy- and peak-to-peak synthesis LMIs~\eqref{eq:syn_LMI_PX_pos_def},~\eqref{eq:syn_LMI_stability},~\eqref{eq:syn_LMI_e2p_p2p}] hold with $(\Psyn,\Ksyn,\Lsyn,\Msyn,\Nsyn)=(\Psyn^\mathrm{old}, \Ksyn^\mathrm{old},\Lsyn^\mathrm{old},\Msyn^\mathrm{old},\Nsyn^\mathrm{old})$.
        \item Suppose the $\mathcal H_\infty$-synthesis LMIs~\eqref{eq:syn_LMI_PX_pos_def},~\eqref{eq:syn_LMI_e2e} [or energy- and peak-to-peak synthesis LMIs~\eqref{eq:syn_LMI_PX_pos_def},~\eqref{eq:syn_LMI_stability},~\eqref{eq:syn_LMI_e2p_p2p}] hold. 
        Then, the $\mathcal H_\infty$-analysis LMIs\eqref{eq:stab_LMI2},~\eqref{eq:ETE_LMI} [or energy- and peak-to-peak analysis LMIs~\eqref{eq:stab_LMI1},~\eqref{eq:stab_LMI2},~\eqref{eq:ana_LMI2} with~\eqref{eq:M1M2_restriction}] hold for $\P$ defined in~\eqref{eq:Pdef} based on $\hat \P$ and for $\K$ and $\hat \P$ constructed by finding square non-singular matrices $\Vsyn$ and $\Usyn$ satisfying $I-\Xsyn\Ysyn = \Usyn \Vsyn^\top$ and computing
    \begin{align}\label{eq:Phat_trafo}
        \hat\P&= \begin{pmatrix}\Ysyn & \Vsyn \\ I & 0\end{pmatrix}^{-1} \begin{pmatrix} I & 0 \\ \Xsyn & \Usyn\end{pmatrix} \\ 
        \begin{pmatrix} \A{\K_\rho} & \B{\K_\rho}{} \\ \C\K{} & \D\K{}{} \end{pmatrix} &=
        \begin{pmatrix}\Usyn & \Xsyn\B{\genplantrho}\u \\ 0&I \end{pmatrix}^{-1} \begin{pmatrix} \Ksyn - \Xsyn \A{\genplantrho}\Ysyn & \Lsyn \\ \Msyn & \Nsyn \end{pmatrix} \begin{pmatrix}\Vsyn^\top & 0 \\ \C\genplant\y \Ysyn & I\end{pmatrix}^{-1}.\label{eq:Kdef}
    \end{align}
    \end{enumerate}
\end{theorem}
\begin{proof}
    Due to Theorem~\ref{thm:warm_start} we can work equivalently with the $\hat\ $ versions of the analysis LMIs which have the unified form~\eqref{eq:standard_LMI_form}.
    First, due to the special structure of $\hat \filt$ and $\hat \M$ from Theorem~\ref{thm:fact} we know that $\hat \Sigma_\rho$ has the following structure and can be transformed by $T$ 
    \begin{align}
        \blkmat{c:c:c}{
            \A{\allhatrho} & \B{\allhatrho}\p & \B{\allhatrho}\w \newblkdash[2.5ex] 
            \C{\allhat}{\hat\filty} & \D{\allhat}{{\hat\filty}}\p & \D{\allhat}{{\hat\filty}}\w \newblkdash
            \C{\hat\all}\z & \D{\all}\z\p & \D{\all}\z\w 
        } &= \left(\begin{array}{ccc:c:c}
            \hat A_1 & 0 & \hat B_1 \C\GK\q & \hat B_1 \D\GK\q\p &\hat B_1 \D\GK\q\w\\ 
            0 & \hat A_2 & 0 & \hat B_2 & 0 \\
            0 &0 & \A\GKrho & \B\GKrho\p & \B\GKrho\w \newblkdash
            \hat C_{11} & \hat C_{12} & \hat D_{11}\C\GK\q & \D{\hat \all}{\hat\filty_1} \p& \hat D_{11}\D\GK\q\w \\
            0 & \hat C_{22} & 0 & \hat D_{22} & 0 \newblkdash
            0 & 0 & \C\GK\z & \D\GK\z\p & \D\GK\z\w
        \end{array}\right), \qquad \label{eq:trafo_inv}
        T\triangleq \begin{pmatrix}I_{\nfiltone} &0 &0 & 0 & 0 \\ 0 & I_{\nfilttwohat} & 0 & 0& 0\\ 0& 0& I_{\nx+\nK} & 0 & 0 \\  0 & -\hat D_{22}^{-1} \hat C_{22} & 0 & \hat D_{22}^{-1} & 0 \\ 0 & 0 & 0& 0& I_{\nw}\end{pmatrix}
    \end{align}
    with $\D{\hat \all}{\hat{\filty}_1} \p\triangleq \hat D_{12} + \hat D_{11} \D\GK\q\p$ to
    \begin{align}\label{eq:trafo_inv_sys}
        \blkmat{c:c:c}{
            \A{\allinv_\rho} & \B{\allinv_\rho}{\hat\filty_2} & \B{\allinv_\rho}\w \newblkdash[2.5ex]
            \C{\allinv}{\hat\filty_1} & \D{\allinv}{\hat\filty_1}{\hat\filty_2} & \D{\allinv}{\hat\filty_1}\w\\
            0& I& 0 \newblkdash
            \C{\allinv}\z & \D{\allinv}\z{\hat\filty_2} & \D{\allinv}\z\w 
        }
        \triangleq
        \blkmat{c:c:c}{
            \A{\allhatrho} & \B{\allhatrho}\p & \B{\allhatrho}\w \newblkdash[2.5ex]
            \C{\allhat}{\hat\filty} & \D{\allhat}{{\hat\filty}}\p & \D{\allhat}{{\hat\filty}}\w \newblkdash
            \C{\hat\all}\z & \D{\all}\z\p & \D{\all}\z\w 
        } T 
        =\blkmat{ccc:c:c}{
            \hat A_1 & -\hat B_1 \D\GK\q\p \hat D_{22}^{-1} \hat C_{22} & \hat B_1 \C\GK\q & \hat B_1 \D\GK\q\p\hat D_{22}^{-1} &\hat B_1 \D\GK\q\w\\[0.02cm]
            0 & \hat A_2-\hat B_2\hat D_{22}^{-1} \hat C_{22} & 0 & \hat B_2 \hat D_{22}^{-1} & 0 \\[0.02cm]
            0 &-\B\GKrho\p \hat D_{22}^{-1} \hat C_{22} & \A\GKrho & \B\GKrho\p \hat D_{22}^{-1} & \B\GKrho\w \newblkdash[2.5ex]
            \hat C_{11} & \hat C_{12}-\D{\hat\all}{\hat\filty_1}\p \hat D_{22}^{-1} \hat C_{22} & \hat D_{11}\C\GK\q &\D{\hat\all}{\hat\filty_1}\p \hat D_{22}^{-1} & \hat D_{11}\D\GK\q\w \\
            0 & 0 & 0 & I_{\np} & 0 \newblkdash
            0 & -\D\GK\z\p\hat D_{22}^{-1} \hat C_{22} & \C\GK\z & \D\GK\z\p\hat D_{22}^{-1} & \D\GK\z\w 
        }.
    \end{align}
    Hence, inequality~\eqref{eq:standard_LMI_form} is equivalent to $T^\top \eqref{eq:standard_LMI_form} T$ which can with $\hat \M = \diag(I_\nq,-I_\np)$ be stated as
    \begin{align}\label{eq:lmi_gen_plant}
        \symb{{\arraycolsep=1pt
            \begin{pmatrix}
                (1-c_1)\hat{\ubar\X}-\hat \P\\
                &-c_1 I \\
                &&c_3 I \\
                &&&c_1 \hat R \\
                &&&& c_1 I\\
                &&&&&c_2 I\\
            \end{pmatrix}
            }\hspace{-0.1cm}
            \begin{pmatrix}
                I & 0 & 0\\
                0 & I & 0\\
                0 & 0 & I\\
                \A{\allinvrho} & \B{\allinvrho}{\hat\filty_2} & \B{\allinvrho}\w \\
                \C{\allinv}{\hat\filty_1} & \D{\allinv}{\hat\filty_1}{\hat\filty_2} & \D{\allinv}{\hat\filty_1}\w \\
                \C{\allinv}\z & \D{\allinv}\z{\hat\filty_2} & \D{\allinv}\z\w 
            \end{pmatrix}
        }\prec 0.
    \end{align}
    In the case of~\eqref{eq:stab_LMI1} or~\eqref{eq:ETE_LMI} where $c_1=1$, $c_2\geq 0$, and $\hat R=\hat \P$ we immediately obtain from the upper left block of~\eqref{eq:lmi_gen_plant} that $-\hat\P + \A{\allinvrho}^\top \hat\P\A{\allinvrho} \prec - \C{\allinv}{\hat\filty_1\top}  \C{\allinv}{\hat\filty_1} - c_2  \C{\allinv}{z\top} \C{\allinv}{\hat\filty_1} \preceq 0$.
    Hence, $\hat\P\succ 0$ if and only if $\A{\allinvrho}$ is Schur stable.
    Theorem~\ref{thm:e2e} and~\ref{thm:p2p} show that under the assumptions of statement (i) $\GK$ is $\ell_{2,\rho}$-stable and thus $\A\GKrho$ is Schur stable.
    Since also $\hat A_1$ as well as $\hat A_2-\hat B_2 \hat D_{22}^{-1}\hat C_{22}$ are Schur stable due to Theorem~\ref{thm:fact}, we conclude that $\hat\P\succ 0$.
    It is straightforward to verify $\Omega_\rho = \genplantrho\star\K_\rho$, and thus we know from~\cite[Section~4.2.2]{Scherer2005} that the transformation~\eqref{eq:P11U},~\eqref{eq:P22V},~\eqref{eq:KLMN} guarantees 
    \begin{align}\label{eq:Pscript}
      \Psyn &= \Ysyntrafo^\top \hat \P \Ysyntrafo\\ \label{eq:ABCDscript}
      \begin{pmatrix}\Asyn &\Bsyn i\\ \Csyn j & \Dsyn ji\end{pmatrix} &= \begin{pmatrix} \Ysyntrafo^\top \hat \P & 0 \\ 0 & I\end{pmatrix} \begin{pmatrix}\A\allinvrho & \B\allinvrho i \\ \C\allinv j & \D\allinv ji\end{pmatrix} \begin{pmatrix} \Ysyntrafo & 0 \\ 0 & I\end{pmatrix}
    \end{align}
    with $\Ysyntrafo=\begin{pmatrix}\Ysyn & I\\ \Vsyn^\top & 0 \end{pmatrix}$.
    Thus, due to~\eqref{eq:Pscript} we have $\Psyn\succ 0$ and due to~\eqref{eq:ABCDscript} we have
    \begin{align}\label{eq:help1}
      &\!\!\begin{pmatrix}I & 0\\
        0 & I \\
        \A{\allinv_\rho} & \B{\allinv_\rho}i \\
        \C{\allinv}{j} & \D{\allinv}ji\\
      \end{pmatrix} \!\diagmat{\Ysyntrafo \\&I} = \begin{pmatrix} \Ysyntrafo \\& I \\ &&\hat\P^{-1}\Ysyntrafo^{-\top}\!\!\!\! \\&&& I \end{pmatrix}
      \!\begin{pmatrix} I&0\\ 0&I \\ 
        \Asyn & \Bsyn i \\ \Csyn j & \Dsyn ji 
      \end{pmatrix}\hspace{-1cm}&
    \end{align}
    for $i\in\{\hat \filty_2,\w\}$ and $j\in\{\hat\filty_1,\z\}$.
    Moreover,~\eqref{eq:P11U},~\eqref{eq:P22V} imply $\hat \P \Ysyntrafo = \begin{pmatrix} I&0\\ \Xsyn & \Usyn \end{pmatrix}^\top$ and thus, we have
    \begin{align*}
      \Ysyntrafo^{-1} &\hat\P^{-1}\hat{\ubar X} \hat \P^{-1} \Ysyntrafo^{-\top} = \symb{
      \blkmat{cc:c}{\hat X&&\\
      &0_{\nx}&\newblkdash 
      \ &&0_{\nK}} \blkmat{c:c}{ I & 0 \newblkdash \Xsyn & \Usyn}^{-1}} = \diagmat{\hat X\\&0_{\nx+\nK}} = \ubar{\hat X}.
    \end{align*}
    Hence, $\mathcal R\triangleq\Ysyntrafo^{-1} \hat\P^{-1}\hat{R} \hat \P^{-1} \Ysyntrafo^{-\top} \in \{\mathcal P^{-1}, \hat{\ubar X}\}$ as $\hat R\in\{\hat\P, \hat{\ubar X}\}$.
    Now we are prepared to multiply~\eqref{eq:lmi_gen_plant} from right by $\diag(\Ysyntrafo, I)$ and from left by its transpose which results with~\eqref{eq:Pscript} and~\eqref{eq:help1} in
    \begin{align}\label{eq:lmi_h3}
            \symb{{\arraycolsep=1pt
                \begin{pmatrix}
                    -\Psyn\\
                    &-c_1 I \\
                    &&c_3 I \\
                    &&&(1-c_1)\hat{\ubar\X}\\
                    &&&&c_1 \mathcal R \\
                    &&&&& c_1 I\\
                    &&&&&& c_2 I\\
                \end{pmatrix}
                }\hspace{-0.1cm}
                \begin{pmatrix}
                    I & 0 & 0\\
                    0 & I & 0\\
                    0 & 0 & I\\
                    \Ysyntrafo 
                    & 0 & 0 \\
                    \Asyn & \Bsyn{\hat \filty_2} & \Bsyn\w \\
                    \Csyn{\hat\filty_1} & \Dsyn{\hat\filty_1}{\hat \filty_2} & \Dsyn{\hat\filty_1}\w \\
                    \Csyn\z & \Dsyn\z{\hat \filty_2} & \Dsyn\z\w
                \end{pmatrix}
            } &\prec 0.
    \end{align}
    If $c_1=1$ and $\mathcal R=\mathcal P^{-1}$, as is the case for~\eqref{eq:stab_LMI1} (with $c_2=0$, $c_3=-\mu$) and~\eqref{eq:ETE_LMI} (with $c_2=\frac 1 \gamma$, $c_3=-\gamma$), then it is straightforward to linearize~\eqref{eq:lmi_h3} with a Schur complement such that we obtain~\eqref{eq:syn_LMI_stability},~\eqref{eq:syn_LMI_e2e}.
    However, for \eqref{eq:ana_LMI2} we have $c_1=\sigma$, $\mathcal R=\hat{\ubar X}$, $c_2=\frac{\alpha}{\gamma}$ and $c_3=-\alpha(\gamma-\beta)$, such that for $\hat X$ that are not positive semi-definite we cannot execute the Schur complement. 
    Therefore, we need to work with the convex relaxation provided by Lemma~\ref{lem:convex_relaxation} in the Appendix~\ref{appendix}.
    Note that $\Ysyntrafo^\top \ubar{\hat{\X}} \Ysyntrafo = \mathcal E_{1} \hat \X \mathcal E_1 $ and $\symbscalar{\ubar{\hat X} \begin{pmatrix} \Asyn & \Bsyn{\hat \filty_2} & \Bsyn\w \end{pmatrix}} = \symbscalar{\hat X \begin{pmatrix} \mathcal E_{2} & \mathcal E_2 & \mathcal E_3 \end{pmatrix} }$.
    Hence, we use Lemma~\ref{lem:convex_relaxation} with $\hat \X=\Zdone-\Zdtwo$, $\Zdone = L^\top_1 L_1$, $\ZX = \begin{pmatrix} \mathcal E_{2}&\mathcal E_{3}&\mathcal E_{4}\end{pmatrix}$, and $\ZY = \begin{pmatrix} \mathcal E_{2}^\mathrm{old}&\mathcal E_{3}^\mathrm{old}&\mathcal E_{4}^\mathrm{old}\end{pmatrix}$ and obtain that~\eqref{eq:lmi_h3} holds if (and for $\mathcal{E}_i = \mathcal E^\mathrm{old}_i$ only if)
    \begin{align}
            \symb{{\arraycolsep=1pt
                \begin{pmatrix}
                    -\Psyn\\
                    &-\sigma I \\
                    &&-\alpha(\gamma-\beta) I \\
                    &&& (1-\sigma) \Zdtwo && -(1-\sigma) \Zdtwo \\ 
                    &&&& \sigma \Zdtwo && -\sigma \Zdtwo \\               
                    &&& \star && (1-\sigma) I\\
                    &&&& \star && \sigma I \\
                    &&&&&&& \sigma I\\
                    &&&&&&&& \frac {\alpha}{\gamma} I\\
                \end{pmatrix}
                }\hspace{-0.1cm}
                \begin{pmatrix}
                    I & 0 & 0\\
                    0 & I & 0\\
                    0 & 0 & I\\
                    \mathcal{E}_{1}^\mathrm{old} & 0 & 0 \\
                    \mathcal{E}_{2}^\mathrm{old} & \mathcal E_{3}^\mathrm{old} & \mathcal{E}_{4}^\mathrm{old} \\
                    L_1\mathcal{E}_{1} & 0 & 0 \\
                    L_1\mathcal{E}_{2} & L_1\mathcal E_{3} & L_1\mathcal{E}_{4} \\
                    \Csyn{\hat\filty_1} & \Dsyn{\hat\filty_1}{\hat \filty_2} & \Dsyn{\hat\filty_1}\w \\
                    \Csyn\z & \Dsyn\z{\hat \filty_2} & \Dsyn\z\w
                \end{pmatrix}
            } &\prec 0.
     \label{eq:lmi_peak_hat4}
    \end{align}
    Applying the Schur complement to~\eqref{eq:lmi_peak_hat4}  yields~\eqref{eq:syn_LMI_e2p_p2p}.
    Similarly, $\hat \P -\ubar{\hat \X} \succ 0$ is transformed to $\Ysyntrafo^\top (\hat \P -\ubar{\hat \X}) \Ysyntrafo = \Psyn -     \Ysyntrafo^\top \hat{\ubar{X}} \Ysyntrafo \succ 0$.
    Again, we need to use Lemma~\ref{lem:convex_relaxation}, now with $\ZX=\mathcal E_{1}$ and $\ZY=\mathcal E_{1}^\mathrm{old}$, to obtain that $\Psyn - \mathcal E_{1}^\top \hat X \mathcal E_{1} \succ 0$ holds if (and for $\mathcal{E}_{1} = \mathcal E_{1}^\mathrm{old}$ only if)
    \begin{align}
        \Psyn+ \symb{\begin{pmatrix}
             - \Zdtwo &  \Zdtwo\\ \Zdtwo & -  I
        \end{pmatrix} \begin{pmatrix}
            \mathcal E_{1}^\mathrm{old}\\
            L_1 \mathcal E_{1} 
        \end{pmatrix}} \succ 0\quad \Leftrightarrow\quad  \eqref{eq:syn_LMI_PX_pos_def}.
    \end{align} 
    Hence, we have shown that the $\mathcal H_\infty$-synthesis LMIs~\eqref{eq:syn_LMI_PX_pos_def},~\eqref{eq:syn_LMI_e2e} [or energy- and peak-to-peak synthesis LMIs~\eqref{eq:syn_LMI_PX_pos_def},~\eqref{eq:syn_LMI_stability},~\eqref{eq:syn_LMI_e2p_p2p}] hold if (and for $(\K,\P)=(\K^\mathrm{old},\P^\mathrm{old})$ only if) the $\mathcal H_\infty$-analysis LMIs~\eqref{eq:stab_LMI2},~\eqref{eq:ETE_LMI} [or energy- and peak-to-peak analysis LMIs~\eqref{eq:stab_LMI1},~\eqref{eq:stab_LMI2},~\eqref{eq:ana_LMI2} with~\eqref{eq:M1M2_restriction}] hold.
    Furthermore, the new controller $\K$ and $\hat \P$ can be reconstructed by~\eqref{eq:Kdef},~\eqref{eq:Phat_trafo} as shown in~\cite[Section~4.2.2]{Scherer2005}. 
The corresponding $\P$ can be computed by~\eqref{eq:Pdef}.
This concludes the proof.
\end{proof}

\section{Algorithm to design multi-objective robust controllers}\label{sec:algo}
The proposed procedure to synthesize robust controllers minimizing the energy- or peak-to-peak gain is sketched in Algorithm~\ref{algo}.
As a starting point for this algorithm, we assume that we have given a filter $\filt$, and a set $\MXset$ such that for all $\Delta\in\Deltaset$ the loop-transformed $\Delta_\rho$ satisfies the finite horizon IQC with terminal cost defined by $(\A\filt,\BB\filt,\CC\filt,\DD\filt,\X,\M)$ for all $(\M,\X) \in \MXset$.
If the set $\MXset$ can be described using LMIs, then all steps in Algorithm~\ref{algo} are convex and can be solved using semi-definite programming~\cite{Boyd2004}.
Due to Theorem~\ref{thm:syn}, we know that a solution of the analysis step provides a solution for the synthesis step and vice versa.
Hence, the upper bound $\gamma$ on the desired performance criterion is improved in every step in every step of Algorithm~\ref{algo}, i.e.,
\begin{align*}
    \gamma_\mathrm{a}^{(i)}\geq \gamma_\mathrm{s}^{(i)}\geq \gamma_\mathrm{a}^{(i+1)}.
\end{align*}

\begin{algorithm}[b]
    \caption{Robust energy- and peak-to-peak synthesis}\label{algo}
    \begin{algorithmic}
        \State Given $\sigma \in (0,1)$, $\rho \in (0,1]$, $\filt$, $\G$, $\MXset$ 
        \State Perform nominal synthesis (i.e., $\np=\nq=\nfilt=\nfilty=0$) to obtain $K^{(0)}$
        \For {i=1,...,N-1}
        \State Set $K=K^{(i-1)}$ 
        \State Solve $\big(\M^{(i)},\X^{(i)},\gamma_\mathrm{a}^{(i)},\mu_\mathrm{a}^{(i)},\P^{(i)}\big) = \argmin\limits_{\substack{(\M,\X)\in\mathbb{MX}, \gamma\geq\mu\geq 0, \P=\P^\top\\ \text{s.t.~\eqref{eq:stab_LMI1},~\eqref{eq:stab_LMI2},~\eqref{eq:ana_LMI2},~\eqref{eq:M1M2_restriction}}}}\gamma$
        \State Based on $\filt$, $\M^{(i)}$, $\X^{(i)}$, compute $\hat \filt$, $\hat \M$, $\hat \X$ according to Theorems~\ref{thm:fact} and~\ref{thm:term_cost}
        \State Set $K^\mathrm{old}=K^{(i-1)}$, $\P^\mathrm{old} = \P^{(i)}$
        \State Solve $\big(\Ysyn^{(i)},\Xsyn^{(i)},\Ksyn^{(i)},\Lsyn^{(i)},\Msyn^{(i)},\Nsyn^{(i)},\gamma_\mathrm{s}^{(i)},\mu_\mathrm{s}^{(i)}\big) = \argmin\limits_{\substack{\Ysyn,\Xsyn,\Ksyn,\Lsyn,\Msyn,\Nsyn,\gamma\geq \mu\geq 0\\ \text{s.t.~\eqref{eq:syn_LMI_PX_pos_def},~\eqref{eq:syn_LMI_stability},~\eqref{eq:syn_LMI_e2p_p2p}}}}\gamma$
        \State Based on $\Ysyn^{(i)},\Xsyn^{(i)},\Ksyn^{(i)},\Lsyn^{(i)},\Msyn^{(i)},\Nsyn^{(i)}$, obtain $\K^{(i)}$ from~\eqref{eq:Kdef}.
        \EndFor
        \State Set $K=K^{(i)}$ and solve $\gamma_\mathrm{a}^{(N)} = \min\limits_{\substack{(\M_1,\X_1)\in\mathbb{MX},(\M_2,\X_2)\in\mathbb{MX}, \gamma\geq\mu\geq 0, \P=\P^\top\\ \text{s.t.~\eqref{eq:stab_LMI1},~\eqref{eq:stab_LMI2},~\eqref{eq:ana_LMI2}}}}\gamma$
    \end{algorithmic}
    \end{algorithm}

    If we want to minimize the $\mathcal H_\infty$-norm instead of the energy- or peak-to-peak gain, then we slightly adapt Algorithm~\ref{algo}. 
    In particular, we only need to change the analysis LMIs from~\eqref{eq:stab_LMI1},~\eqref{eq:stab_LMI2},~\eqref{eq:ana_LMI2},~\eqref{eq:M1M2_restriction} to~\eqref{eq:stab_LMI2},~\eqref{eq:ETE_LMI} and the synthesis LMIs from~\eqref{eq:syn_LMI_PX_pos_def},~\eqref{eq:syn_LMI_stability},~\eqref{eq:syn_LMI_e2p_p2p} to~\eqref{eq:syn_LMI_PX_pos_def},~\eqref{eq:syn_LMI_e2e}.
    
    When we want to minimize the maximum of the $\Hinf$-norm and the energy-to-peak gain of $\DGK$, then we apply Algorithm~\ref{algo} but solve the analysis problem subject to~\eqref{eq:stab_LMI1},~\eqref{eq:stab_LMI2},~\eqref{eq:ETE_LMI},~\eqref{eq:ana_LMI2},~\eqref{eq:M1M2_restriction} and the synthesis problem subject to~\eqref{eq:syn_LMI_PX_pos_def},~\eqref{eq:syn_LMI_e2e},~\eqref{eq:syn_LMI_stability},~\eqref{eq:syn_LMI_e2p_p2p} with the energy-to-peak parameters $\rho=1$, $\alpha=1$, $\beta=0$.

    We can even mix different performance goals for several performance channels $\w_{(j)}\to \z_{(j)}$ (compare~\cite{Scherer1997} for nominal multi-objective synthesis). 
    For example, we may want to minimize the sum of the $\mathcal H_\infty$-norm $\gamma_{(1)}$ of $\w_{(1)}\to\z_{(1)}$, the energy-to-peak gain $\gamma_{(2)}$ of $\w_{(2)}\to\z_{(2)}$, and the peak-to-peak gain $\gamma_{(3)}$ of $\w_{(3)}\to\z_{(3)}$.
    Or we may want to minimize $\max\{\gamma_{(1)},\gamma_{(2)}\}$ subject to $\gamma_{(3)}\leq 1$.
    We can solve such problems by minimizing the desired criterion subject to the analysis LMIs for each channel $j$ restricted to a common Lyapunov matrix $\P_{(j)}=\P$ and a common IQC multiplier $\M_{(j)}=\M$.
    We call the analysis with these restrictions a \emph{common analysis} in contrast to an \emph{individual analysis} where each channel is analyzed individually without these restrictions.
    A common analysis using the same $\M$ and $\P$ in all channels certainly introduces conservatism, however, this restriction is necessary such that the synthesis LMIs lead to one controller $\K$ satisfying the desired performance in all channels rather than to several controllers $K_{(j)}$ satisfying only the performance criterion for channel $\w_{(j)}\to\z_{(j)}$.
    When combining peak-to-peak, which needs $\rho<1$, with one of the other two measures, which have $\rho=1$, then we need to ensure that $(\Ksyn,\Lsyn,\Msyn,\Nsyn)$ do not correspond to $\K$ and $\K_\rho$ at the same time.
    One way to circumvent this problem is to pull $\rho$ of $\Ksyn$ and $\Lsyn$ in~\eqref{eq:KLMN} for the peak-to-peak LMIs such that $(\Ksyn,\Lsyn,\Msyn,\Nsyn)$ correspond to $\K$ also in the peak-to-peak LMIs.

We conclude this section with some further remarks on the implementation.
\begin{remark}[(Initialization)]
    The initialization with the nominal controller $K^{(0)}$ may often not be robustly stabilizing leading to infeasibility in step $i=1$.
    In such situation, we rescale $\Delta^\mathrm{new}=\tau \Delta$ with some $\tau \in [0,1]$ to achieve feasibility.
    Then, we iterate as described in Algorithm~\ref{algo} but in every analysis and synthesis step we maximize $\tau$ until we eventually reach $\tau =1$. 
    If we never reach $\tau=1$, then we are not able to provide a robustly stabilizing controller for the uncertain system.
    This happens, for example, if the uncertain system is not robustly stabilizable.
\end{remark}
\begin{remark}[(Convex relaxation)]
    Since $\hat \X$ is generally indefinite, we decompose it via~\eqref{eq:hatX_decomposition} and use a convex relaxation based on a solution $\P^\mathrm{old},\mu^\mathrm{old},\gamma^\mathrm{old}$ of the analysis LMIs for some controller $\K^\mathrm{old}$. 
    This convex relaxation is needed to convexify the synthesis LMIs but causes that (i) of Theorem~\ref{thm:syn} requires $(\K,\P)=(\K^\mathrm{old},\P^{\mathrm{old}})$ and cannot guarantee equivalence between the synthesis and analysis LMIs for all $\K$. 
    In practice, this conservatism can be seen in a possibly slower convergence rate of the algorithm.
    However, if the algorithm has converged and $\Psyn^{(i)}=\Psyn^{(i-1)}$, $\Ksyn^{(i)}=\Ksyn^{(i-1)}$, etc., then, the convex relaxation is tight and does not introduce any conservatism.
    Further, note that the decomposition~\eqref{eq:hatX_decomposition} is not unique, as we can add the same positive semi-definite matrix to $L_1^\top L_1$ and $L_2^\top L_2$ without changing~\eqref{eq:hatX_decomposition}. 
    However, such an addition introduces unnecessary conservatism and hence, we choose $L_1$ and $L_2$ as small as possible in the following sense.
    We set the dimension $L_1\in\R^{n_+ \times \nfilthat}$ and $L_2 \in\R^{n_- \times \nfilthat}$ where $n_+$ and $n_-$ are the numbers of positive and negative eigenvalues of $\hat\X$.
    The decomposition $\hat \X = \symb{\begin{pmatrix} I_{n_+} &0\\ 0 & -I_{n_-}\end{pmatrix} \begin{pmatrix} L_1 \\ L_2 \end{pmatrix}}$ exists and is unique as $\hat \X$ is symmetric and hence diagonalizable by an orthonormal matrix.
\end{remark}
\begin{remark}[(Optimizing $\rho$ and $\sigma$)] \label{rem:rho_opt}
    Instead of fixing the parameters $\sigma\in(0,1)$ and $\rho\in (0,1)$ for peak-to-peak ($\rho=1$ for energy-to-peak) beforehand, one can also optimize over $\rho$ and/or $\sigma$ in each iteration step.
    However, these parameters enter nonlinearly and require a line search, which means that the SDP in each iteration of Algorithm~\ref{algo} needs to be solved several times to find the best (or a better) $\rho$ and $\sigma$.
    Furthermore, note that we can easily derive synthesis LMIs for $\sigma=0$ or $\sigma=1$ as well.
    In particular, if $\sigma=0$ or $\sigma=1$, then we can first cancel the $0$s in~\eqref{eq:lmi_peak_hat4} and then apply the Schur complement resulting in an LMI with reduced order, which we can use instead of the problematic LMI~\eqref{eq:syn_LMI_e2p_p2p} containing $\frac{1}{\sigma}$ and $\frac{1}{1-\sigma}$.
\end{remark}
\begin{remark}[(Controller order $\nK$)]
    The order of the controller is $n_{K^{(i)}}=\nx+n_{\hat\filt^{(i)}}$. 
    As the order $n_{\hat\filt^{(i)}}$ of the filter $\hat\filt^{(i)}$ from the factorization in Theorem~\ref{thm:fact} depends on $M^{(i)}$, it may change over the iterations $i$.
    Hence, also $K^{(i-1)}$ and $K^{(i)}$ may have different orders. 
    However, we need to use $K^{(i-1)}$ to define $K^{\mathrm{old}}$ for the synthesis step $i$ and we need to have $n_{K^{(i)}}=n_{K^{\mathrm{old}}}$.
    Thus, in the case $n_{K^{(i)}}>n_{K^{(i-1)}}$ we choose a non-minimal state space representation for $K^\mathrm{old}=K^{(i-1)}$ such that $n_{K^{(i)}}=n_{K^{\mathrm{old}}}$.
    In the case $n_{K^{(i)}}<n_{K^{(i-1)}}$, we can increase $n_{K^{(i)}}=\nx+n_{\hat\filt^{(i)}}$ by introducing unobservable and uncontrollable stable modes into the filter $\hat\filt^{(i)}$ to ensure $n_{K^{(i)}}=n_{K^\mathrm{old}}$.
    Instead of increasing the order of $\hat \filt^{(i)}$, one can also reduce the order of $K^{\mathrm{old}}$, for example, by using \textsc{Matlab}'s \texttt{reducespec} and \texttt{getrom}, which leads to smaller LMI sizes and faster computation times.
    However, this model order reduction does not guarantee feasibility of the warm start anymore.
    Nevertheless, we observed in all our tests that the controller order of $K^\mathrm{old}$ can be reduced from $\nx + n_{\hat\filt^{(i-1)}}$ to $\nx + n_{\hat\filt^{(i)}}$ without any loss in the performance.
\end{remark}

\section{Examples}\label{sec:exmp}

This section demonstrates the proposed algorithm for analysis and robust controller synthesis in three examples. 
Example~\ref{exmp1} compares the proposed analysis procedure of the energy-to-peak and peak-to-peak gain via IQCs to related work.
We compare the synthesis with respect to different performance criteria in Example~\ref{exmp2} and show how to perform a multi-objective synthesis in Example~\ref{exmp3}.
All examples are implemented in Matlab using YALMIP~\cite{Lofberg2004} and MOSEK~\cite{mosek}. The code is available online\footnote{\url{https://github.com/Schwenkel/multi-objective-iqc-synthesis}}.

\begin{example}[Analysis of the energy- and peak-to-peak gain]\label{exmp1}
    In this example, we compare Theorem~\ref{thm:p2p} to related results.
    The peak-to-peak analysis result presented in~\cite[Theorem~4]{Schwenkel2023b} is the special case of Theorem~\ref{thm:p2p} with the restriction~\eqref{eq:M1M2_restriction} and $\sigma=0$.
    With this method and the IQC for parametric time-varying uncertainties from~\cite[Theorem~4]{Schwenkel2023b} a value of $\gamma = 67.81$ was achieved for~\cite[Example~14]{Schwenkel2023b}.
    Choosing $\sigma = 0.6$ we can improve this bound to  $\gamma = 67.09$ and by dropping the restriction~\eqref{eq:M1M2_restriction} completely, we achieve $\gamma= 66.93$.
    This shows that~\eqref{eq:M1M2_restriction} indeed introduces some conservatism, but with a proper choice of $\sigma$ the conservatism is rather small.
    As second example, we consider the one from~\cite[Section~IV.~B]{Jaoude2020} where the worst-case energy-to-peak gain is bounded by $\gamma= 2.683$ via IQCs using results from~\cite{Fry2017}. 
    Using Theorem~\ref{thm:p2p}, we can significantly improve this bound to $\gamma = 2.008$.
\end{example}

The uncertainties that we consider in the following examples are structured uncertainties consisting of parametric and dynamic uncertain components. 
Therefore, we first recap the finite horizon IQCs with terminal cost for these components.

\emph{Uncertainty description via IQCs.}
A parametric uncertainty $\delta\in [\ubar \delta,\bar\delta]$, $\ubar \delta < 0 < \bar\delta$ is characterized by the discrete-time analog of the IQC from~\cite[Theorem~12]{Scherer2018}.  
Hence, for $\psi = \ssrep{\A\psi}{\BB\psi}{\CC\psi}{\DD\psi}\in\RHinf^{\np\times \np}$ the uncertainty $\p=\delta \q$ satisfies the finite horizon IQC defined by $\filt=\filt_\text{p} = \begin{pmatrix}
    \bar\delta & -1\\ -\ubar \delta & 1
\end{pmatrix} \otimes \psi $ with state space representation
    \begin{align*}
         \blkmat{c:c}{
            \A{\filt_\text{p}} & \BB{\filt_\text{p}} \newblkdash \CC{\filt_\text{p}} & \DD{\filt_\text{p}} 
        } = \blkmat{cc:cc}{\A\psi & 0 & \bar\delta \BB\psi & -\BB\psi \\ 0 & \A\psi & -\ubar \delta \BB\psi & \BB\psi \newblkdash \CC\psi & 0 & \bar\delta \DD\psi & -\DD\psi \\ 0 & \CC\psi & -\ubar \delta \DD\psi & \DD\psi},
    \end{align*}
    and $(\M,\X)\in \MXset_\text{p}$ where $\MXset_\text{p}$ is defined as the set of all $\M = \begin{pmatrix}
        0 & \M_{12} \\ \M_{12}^\top & 0
    \end{pmatrix}$ and $\X = \begin{pmatrix}
        0 & \X_{12} \\ \X_{12}^\top & 0
    \end{pmatrix}$ for which there exists $R$ such that $R-X \prec 0$ and
    \begin{align*}
        \symb{\diagmat{-R \\ &R \\ &&\M } \begin{pmatrix} I & 0 \\ \A{\filt_\text{p}} & \BB{\filt_\text{p}} E\\ \CC{\filt_\text{p}} & \DD{\filt_\text{p}} E
        \end{pmatrix}} \succ 0
    \end{align*}
    with $E=\begin{pmatrix} I & 0 \end{pmatrix}^\top$.
    We choose $\psi= \ssrep{\A\psi}{\BB\psi}{\CC\psi}{\DD\psi}=\begin{pmatrix}
        I_{\np} & \frac{\eins_{\np}}{\zs - a} & \dots & \frac{\eins_{\np}}{(\zs-a)^\greeknu}
    \end{pmatrix}^\top$ with parameters $a\in (-1,1)$ and $\greeknu \in \N$ and where $\eins_\np \in \R^{\np\times \np}$ is the all-ones-matrix.
    
    Furthermore, a scalar dynamic uncertainty $\Delta\in \Hinf^{1\times 1}$ with $\|\Delta\|_\infty < \gamma_\Delta$ is characterized by the discrete-time analog of the IQC from~\cite[Theorem~21]{Scherer2022a}.
    Hence, $\p=\Delta (\q)$ satisfies the finite horizon IQC defined by $\filt =\filt_\text{d} = I_2 \otimes \psi$ and $(\M,\X)\in \MXset_\text{d}$ where $\MXset_\text{d}$ is defined as the set of all $\M = \diag(\gamma_\Delta , -\frac{1}{\gamma_\Delta} ) \otimes \M_\text{d}$ and $\X  = \diag( \gamma_\Delta, -\frac{1}{\gamma_\Delta})\otimes \X_\text{d}$ with 
    \begin{align*}
        \symb{\diagmat{-\X_\text{d} \\ &\X_\text{d} \\ &&\M_\text{d} } \begin{pmatrix} I & 0 \\ \A\psi & \BB\psi \\ \CC\psi & \DD\psi
        \end{pmatrix}} \succ 0.
    \end{align*}

\begin{example}\label{exmp2}
    \emph{System.}
    Consider the system $G=\ssrep{\A\G}{\BB\G}{\CC\G}{\DD\G}$ with
    \begin{align*}
        \blkmat{c:c}{\A\G&\BB\G \newblkdash \CC\G & \DD\G} = \blkmat{c:c:c:c}{\A\G&\B\G\p&\B\G\w&\B\G\u \newblkdash
        \C\G\q & \D\G\q\p & \D\G\q\w & \D\G\q\u \newblkdash 
        \C\G\z & \D\G\z\p & \D\G\z\w & \D\G\z\u \newblkdash 
        \C\G\y & \D\G\y\p & \D\G\y\w & \D\G\y\u} =
        \blkmat{cc:cc:cc:cc}{.6 & .2& .2 & .2& 3 & 2 & 1 & 0\\ -.1 & -.3 & .3 & -.2 & 3 & 1 & 2 & .2 \newblkdash .2 & -.3 & .4 & .3 & 3 & 1 & 0 & 0 \\ .8 & .5 & -.6 & .1& 2 & 7 & 0 & .1 \newblkdash 2&1&1 & 2 & 1 & -2 & 4 & 0 \\ 2&3&-1&4&-4&3&0&0\newblkdash 
        1&0&0&0&.1&.2&0&0}.
    \end{align*}
    The uncertainty $\p=\Delta \q$ is parametric $\Delta = \begin{pmatrix}
        \delta_1 & 0 \\ 0 & \delta_2
    \end{pmatrix}$ with $\delta_1 \in [-0.1,0.5]$ and $\delta_2 \in [-0.3, 0.6]$.
    Therefore, we use the IQC for parametric uncertainties $\filt_\text{p}$, $\MXset_\text{p}$ for both uncertain parameters and stack it.
    We choose the parameters $a=-0.25$ and $\greeknu=4$ to define $\psi$.
    
    \emph{Controller design.}
    Using the proposed Algorithm~\ref{algo} with $\sigma=0.95$ and optimizing $\rho$ as described in Remark~\ref{rem:rho_opt}, we design the controllers $\K_{\Hinf}$, $\K_\ETP$, and $\K_\PTP$ for the $\mathcal H_\infty$, energy-, and peak-to-peak performance, respectively.
    Table~\ref{tab:exmp1} shows the resulting gains for these controllers where the upper bounds are provided by Theorems~\ref{thm:e2e} and~\ref{thm:p2p}. 
    The lower bounds are obtained by maximizing the performance over the uncertainty $\Delta$ and disturbance $d$. 
    Since finding the true worst-case $\w$ and $\Delta$ is a non-convex optimization problem, we can only find local maxima.
    Table~\ref{tab:exmp1} shows that our method yields tight bounds on the $\mathcal H_\infty$-norm and energy-to-peak gain, as lower and upper bounds (almost) coincide. 
    Although the gaps between the lower and upper bound indicate the possibility of some conservatism in the bound on peak-to-peak gain, it is still worthwhile minimizing this upper bound as $\K_\PTP$ achieves a significantly better peak-to-peak performance than the other controllers.
    
    \emph{Comparison to $\mu$-synthesis.} If we want to optimize the $\mathcal H_\infty$-performance, we can also use \textsc{Matlab}'s \texttt{musyn}~\cite{Balas1991} which uses the structured singular value $\mu$ instead of IQCs to characterize the uncertainty $\Delta$.
    Using $\mu$-synthesis we obtain a controller with $\mathcal H_\infty$-performance $\gamma=37.65$, which is marginally larger than the $\gamma=37.47$ we obtained via IQC synthesis.
    We take this example as an indication that the proposed IQC synthesis can provide equally good results as $\mu$-synthesis.
    Note that compared to the proposed IQC synthesis procedure, $\mu$-synthesis can only be applied to the $\Hinf$-problem with linear time-invariant uncertainties.
    \begin{table}[htb]
        \centering
        \begin{tabular}{cccccc}
            \toprule
            controller & $N$ & time$/N$ & $\mathcal H_\infty$-norm &  energy-to-peak gain &  peak-to-peak gain \\ \midrule
            $\K_{\Hinf}$ & 10 & $7.64$\,s & $37.46\leq \gamma\leq 37.47$ & $35.74 \leq \gamma \leq 35.74$ & $59.73\leq\gamma\leq 68.73$  \\
            $\K_\ETP$ & 20 & $11.12$\,s &  $42.81 \leq \gamma \leq 42.88$ & $34.07 \leq \gamma \leq 34.07$ & $57.27 \leq \gamma \leq 64.28$ \\
            $\K_\PTP$ & 13 & $18.13$\,s &  $46.27 \leq \gamma \leq 46.29$ & $34.30 \leq \gamma \leq 34.41$ & $49.69 \leq \gamma \leq 54.30$ \\ \bottomrule
        \end{tabular}
        \caption{The number of iterations $N$ used to compute the controllers $\K_{\Hinf}$, $\K_\ETP$, $\K_\PTP$ and the average computation time per iteration.
        Further, lower and upper bounds on the $\mathcal H_\infty$-norm, energy- and peak-to-peak gain of these controllers rounded to 4 significant digits.}
        \label{tab:exmp1}
    \end{table}
    
\end{example}
\begin{example}[Multi-objective robust synthesis]\label{exmp3}
    \emph{System.} Consider the example from~\cite{Bemporad1998} of a servomechanism consisting of a DC motor, a gearbox, an elastic shaft, and an uncertain load. The dynamics of this system are
    \begin{align*}
        \dot x = A_\text{c} x + B_{\text c} u_\text{real} = \begin{pmatrix}
            0 & 1 & 0 & 0 \\
            -\frac{1280.2}{J_L} & -\frac{25}{J_L} & \frac{64.01}{J_L} & 0 \\ 0 & 0 & 0 & 1 \\ 128.02 & 0 & -6.4010 & - 10.2
        \end{pmatrix} x + \begin{pmatrix}
            0 \\ 0 \\ 0 \\ 220
        \end{pmatrix} u_\text{real}
    \end{align*}
    where the state $x=\begin{pmatrix}
        \theta_L & \dot \theta_L & \theta_M & \dot \theta_M
    \end{pmatrix}$ is comprised of the load angle $\theta_L$ and the motor angle $\theta_M$.
    The load inertia $J_L\in [5,15]$ is uncertain.
    Hence, we know $\frac{1}{J_L} = \frac{4}{30} + \delta$ with $\delta \in \left[-\bar \delta,\bar \delta\right]$, $\bar \delta \triangleq \frac{2}{30}$ and thus we write $A_\text{c} = A_{\text{c}0} + \delta A_{\text{c}\delta}$ where both $A_{\text{c}0}$ and $A_{\text{c}\delta}$ are independent of the uncertainty.
    The model is discretized with sampling rate $T_s = 0.1$ as \begin{align*}
        \exp\left(\begin{pmatrix}
        T_s(A_{\text{c}0}+\delta A_{\text{c}\delta}) & T_s B_\text{c} \\ 0 & 0
    \end{pmatrix} \right)=\begin{pmatrix}
        \A\G & \B\G\u \\ 0 & I 
    \end{pmatrix}+\delta \begin{pmatrix}
        \C\G{\q_1} & \D\G{\q_1}\u \\ 0 & 0 
    \end{pmatrix} + \begin{pmatrix}
        \mathcal{O}(\delta^2\T_s^2) \\ 0
    \end{pmatrix}
    \end{align*}
    where we neglect the terms $\mathcal O(\delta^2\T_s^2)$. 
    Furthermore, we consider a multiplicative input uncertainty $u_\text{real}=(1+\Delta_\text{u})u$ with $\|\Delta_\text{u}\|_\infty \leq 0.1$ and define the corresponding uncertainty channel by $q_2 = u$ and $p_2 = \Delta_\text{u} q_2$ with $\B\G{\p_2}=\B\G\u$. 
    The uncertainty channel corresponding to the uncertain load inertia $J_L$ is defined by $q_1 = \C\G{\q_1} x +  \D\G{\q_1}\u\u_\text{real}$ and $\p_1 = \delta\q_1$ with $\B\G{\p_1} = I$. 
    Due to $u_\text{real} = \u + \p_2$ we have $\D\G{\q_1}{\p_2}= \D\G{\q_1}\u$. 
    The dynamic uncertainty $\Delta_\text{u}$ satisfies the finite horizon IQC defined by $\filt_\mathrm{d}$ and $(\M,\X)\in \MXset_\mathrm{d}$.
    The uncertainty $\Delta =\diag (\delta I_4 , \Delta_\text{u} )$ is characterized by the finite horizon IQC defined by $\filt = \diag(\filt_\mathrm{p},\filt_\mathrm{d})$ and $\M=\diag(\M_\mathrm{p},\M_\mathrm{d})$, $\X=\diag(\X_\mathrm{p},\X_\mathrm{p})$ with $(\M_\mathrm{p},\X_\mathrm{p})\in\MXset_\mathrm{p}$ and $(\M_\mathrm{d},\X_\mathrm{d})\in\MXset_\mathrm{d}$. 
    We choose the parameters $a=0.5$ and $\greeknu=2$ to define $\psi$.
    We measure the reference $\y_1 = \w_1$ and the load angle $\y_2=\x_1+\w_2$ where $\w_2$ is measurement noise.
    Since the measurement noise contains rather high frequencies and the reference contains only small frequencies we introduce the two dynamic weights $W_1(\zs) = 1.1178\frac{(\zs-0.9659)(\zs-0.6916)}{(\zs-0.9998)(\zs-0.8588)}$ and $W_2(\zs) = 66.661 \frac{\zs -0.9775}{\zs+0.3903}$ that multiply the performance inputs $w_1$ and $w_2$, respectively, that is $G_\text{w}=\G \diag(I_5, W_1, W_2, 1)$.

    \emph{Multi-objective goal.}
    The control goal is to minimize the tracking error $z_1 = x_1 - w_1$ while the input $z_2=u$ and the (normalized) torsional torque $z_3 = \begin{pmatrix}
        16.3 & 0 & -0.815 & 0
    \end{pmatrix} x $ must satisfy $\|z_2\|_\peak \leq 1$ and $\|z_3\|_\peak \leq 1$. Compared to~\cite{Bemporad1998} we normalized both the input and the torsional torque to a peak bound of $1$. 
    To achieve this goal, we want to minimize the $\mathcal H_\infty$-norm $\gamma_1$ from $\w$ to $\z_1$ as well as the energy-to-peak gain $\gamma_2$ from $\w$ to $\z_2$ and the energy-to-peak gain $\gamma_3$ from $\w$ to $\z_3$.
    We mix these goals by minimizing $\gamma_1+\gamma_2+\gamma_3$ subject to the LMI constraints we obtain for each performance channel as described in Section~\ref{sec:algo}. 

    \emph{Convergence of the synthesis algorithm.} Figure~\ref{fig:gam_iterations} shows the value of $\gamma_1+\gamma_2+\gamma_3$ after each analysis step over the first 11 iterations of the algorithm for both the common and the individual analysis as described in Section~\ref{sec:algo}.
    The average computation time per iteration on a regular notebook was $12.6$\,s.
    In the first iteration, the algorithm was only feasible for $\tau=0$ in the analysis step and for $\tau =0.375$ in the synthesis step, hence we did not include the first iteration in the plot.
    The algorithm could choose $\tau = 1$ for the remaining iterations and thus provide a valid robust performance bound.
    After $10$ iterations no more noticeable improvement is made.
    As there is no need for the conservatism of the common analysis in the final analysis step, we perform one final individual analysis step to obtain the upper bound $\gamma_1+\gamma_2+\gamma_3 \leq 3.53$ for the synthesized controller.
    \begin{figure}
        \centering
%
%
%
\begin{tikzpicture}

\begin{axis}[%
width=2in,
height=1.4in,
at={(0.758in,0.481in)},
scale only axis,
xmin=1,
xmax=13,
ymin=2.8,
ymax=11,
axis x line*=bottom,
axis y line*=left,
xlabel={Number of iterations},
axis background/.style={fill=white}
]
\addplot [color=myblue, line width=1pt, mark=*, mark options={scale =0.7, fill=myblue}, only marks]
  table[row sep=crcr]{%
  2	10.4393749569531\\
  3	7.03749652295197\\
  4	5.97558337689111\\
  5	5.32901488689088\\
  6	4.89631752071521\\
  7	4.66787893428087\\
  8	4.52946258772666\\
  9	4.431577863621\\
  10	4.36329649144063\\
  11	4.32221074379639\\
};
\addlegendentry{\ $\gamma_1+\gamma_2+\gamma_3$ (common)\ \ \ };

\addplot [color=myred, line width=1pt, mark=o, mark options={scale =0.7}, only marks]
  table[row sep=crcr]{%
  2	8.5594756290688\\
  3	5.1663336508524\\
  4	4.33233860798672\\
  5	3.96624374606433\\
  6	3.77175839548574\\
  7	3.65483818141959\\
  8	3.60436695093915\\
  9	3.57442596858038\\
  10	3.54897024139315\\
  11	3.53486676959564\\
  12	3.53674098547207\\
};
\addlegendentry{\ $\gamma_1+\gamma_2+\gamma_3$ (individual)};
\end{axis}

\end{tikzpicture}%
        \caption{The value of $\gamma_1+\gamma_2+\gamma_3$ after each (common) analysis step of the algorithm. After each iteration, we performed an additional individual analysis step for comparison.}
        \label{fig:gam_iterations}
    \end{figure}
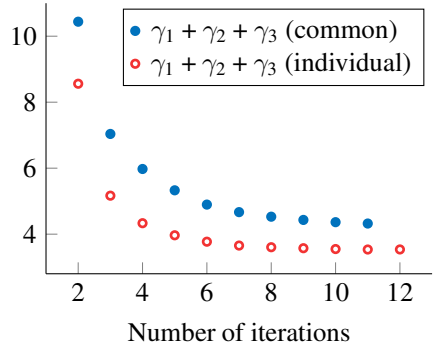

    \emph{Controller order.}
    The order of the synthesized controller is $\nK = \nx + \nfilt = 17$ where $\nx=7$ is the dimension of the system $\G$ and the weights $W_1$ and $W_2$ and $\nfilt=\np\greeknu = 10$ is the dimension of the IQC filter. 
    Such a high order of $\K$ is not necessary, using MATLAB's \texttt{R=reducespec(K,"balanced")} and \texttt{K=getrom(R, Order=7)} we can reduce the order of $\K$ to $\nK=7$ without any loss in performance as a robust performance analysis of the reduced order controller reveals $\gamma_1+\gamma_2+\gamma_3 \leq 3.52$.
    
   \emph{Step response.} Figure~\ref{fig:step_response} shows the step response of the resulting closed loop for a reference of $\w_1 = 90^\circ$ and white measurement noise $\w_2$ with $0^\circ$ mean and $2^\circ$ standard deviation.
    We observe a fast and robust reference tracking performance with no overshoot for all (depicted) $\Delta \in \Deltaset$ as well as constraint satisfaction for all times.
    For comparison, we show in Figure~\ref{fig:step_response} also the step responses when using the nominal controller that was designed with the same mixed synthesis goals but without considering the uncertainties, i.e., with $\Delta = 0$. 
    Although the nominal controller nominally performs even better and has a larger distance to the constraint limits, we see that the performance is not robust against uncertainties $\Delta\in\Deltaset$ and that the constraints are violated for some $\Delta\in\Deltaset$. 
    \begin{figure}[tb]
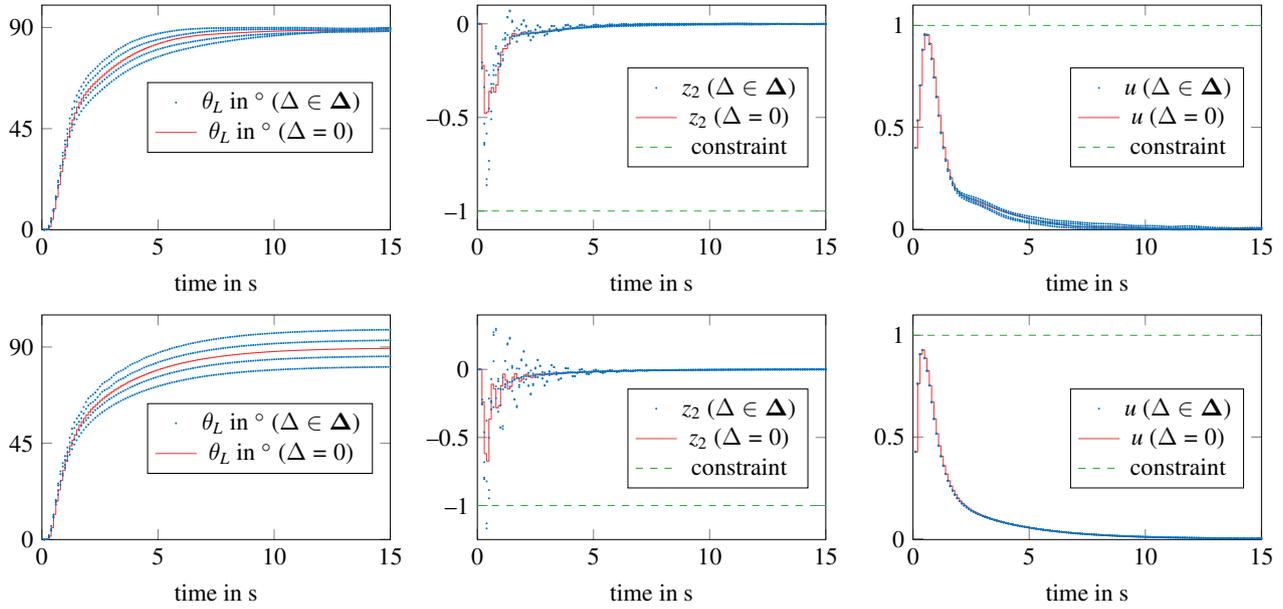

        \centering
        \resizebox{0.95\linewidth}{!}{\input{step_responses_robust}}\\[0.2cm]
        \resizebox{0.95\linewidth}{!}{\input{step_responses_nominal}}
        \caption{Step responses of the synthesized robust controller (upper plot) and the nominal controller (lower plot) for different uncertainties $\Delta\in\Deltaset$. \phantom{mmmmmmmm mmmmmmmmm mmmmmmmmm m}}
        \label{fig:step_response}
    \end{figure}
    
\end{example}

\section{Conclusion}
We have proposed a framework to optimize controllers for uncertain discrete-time systems with respect to various performance measures including $\mathcal H_\infty$, energy-to-peak, and peak-to-peak performance and a mixture thereof.
The framework can handle various types of structured and unstructured uncertainties due to the characterization of the uncertainty with IQCs. 
This work opens new possibilities for robust controller design in uncertain discrete-time systems.
For example, the peak-to-peak minimization procedure can be used to design pre-stabilizing controllers for tube-based MPC algorithms like~\cite{Schwenkel2022a} that minimize the tube size.
Another interesting continuation of this work is to develop energy-to-peak and peak-to-peak minimization via IQCs in continuous time.

\bmsection*{Data availability statement}
The code and data are available online \url{https://github.com/Schwenkel/multi-objective-iqc-synthesis}.

\bmsection*{Acknowledgments}
Funded by Deutsche Forschungsgemeinschaft (DFG, German Research Foundation) under Germany’s Excellence Strategy – EXC 2075 – 390740016. 
C. W. Scherer and F. Allgöwer acknowledge the support by the Stuttgart Center for Simulation Science (SimTech).
L. Schwenkel thanks the International Max Planck Research School for Intelligent Systems (IMPRS-IS) for supporting him. 

\bmsection*{Conflict of interest}

The authors declare no potential conflict of interest.

\bibliography{my_bib}

\appendix

\bmsection{Technical Lemmas}\label{appendix}
\vspace*{12pt}

\begin{lemma}\label{lem:dare}
  For $\filt=\ssrep{A}{B}{C}{D}$ and $\M=\M^\top$ with suitable dimensions assume that $(A,B)$ is controllable, and that $\filt^*\M \filt \succ 0$ holds on $\partial\mathbb{D}$. 
  Let $\begin{pmatrix}Q & S \\ S^\top & R \end{pmatrix}=\symb{\M \begin{pmatrix}C & D \end{pmatrix}}$. 
  Then there exists a unique stabilizing solution $Z_\mathrm{s}\in \textup{\texttt{dare}}(A,B,Q,R,S)$ which satisfies 
  \begin{align*}
    \lambda \left( A+BK_\mathrm{s}\right) \subseteq \mathbb D,\quad K_\mathrm{s}=-(B^\top Z_\mathrm{s} B + R)^{-1}(A^\top Z_\mathrm{s} B + S)^\top
  \end{align*}
  as well as a unique unmixed solution $Z_\mathrm{u}\in \textup{\texttt{dare}}(A,B,Q,R,S)$ which satisfies 
  \begin{align*}
    \lambda \left( A+BK_\mathrm{u}\right) \subseteq \Lambda_\mathrm{u},\quad K_\mathrm{u}=-(B^\top Z_\mathrm{u} B + R)^{-1}(A^\top Z_\mathrm{u} B + S)^\top
  \end{align*}
  with $\Lambda_{\mathrm{u}} = \{0\} \cup \{\lambda \in\mathbb{C} \mid |\lambda|>1\}$.
\end{lemma}
\begin{proof}
    Note that both $\overline {\mathbb{D}}$ and $\overline {\Lambda_\mathrm{u}}\triangleq \{0\} \cup \{\lambda \in\mathbb{C} \mid |\lambda|\geq 1\}$ are unmixed sets (see~\cite{Clements2003} for a Definition of unmixed sets).
    Since $(A,B)$ is controllable and $\filt^*\M \filt \succ 0$ on $\partial \mathbb{D}$, we can thus apply \cite[Theorem 1.1]{Clements2003} to obtain unique solutions $Z_\mathrm{s}$ and $Z_\mathrm{u}$ which satisfy $\lambda(A+BK_\mathrm{s})\subseteq \overline {\mathbb{D}}$ and $\lambda(A+BK_\mathrm{u})\subseteq \overline{\Lambda_\mathrm{u}}$.
    To show that the eigenvalues are not on the unit circle $\partial \mathbb D$ we take a look at $\Pi\triangleq \filt^*\M \filt$.
    Due to $\Pi = \filt^*\M \filt \succ 0$ on $\partial \mathbb{D}$, $\Pi$ has no zeros on $\partial \mathbb{D}$ and hence $\Pi^{-1}$ has no poles on $\partial \mathbb{D}$.
    The descriptor form of $\Pi^{-1}$ is given by (see~\cite[Appendix A]{Hu2017} for a derivation)
    \begin{align*}
        \Pi^{-1}: \quad E_{\Pi^{-1}}\pi_{t+1}  &= \A{\Pi^{-1}} \pi_{t} + \begin{pmatrix}0\\0\\-I\end{pmatrix} u_{t} \\  y_{t} &= \begin{pmatrix}0&0&I\end{pmatrix}\pi_{t} 
    \end{align*}
    where 
    \begin{align*}
        E_{\Pi^{-1}} \triangleq \begin{pmatrix} I & 0 & 0 \\ 0 & - A^\top & 0 \\ 0 & -B^\top & 0\end{pmatrix}, \qquad 
        A_{\Pi^{-1}} \triangleq \begin{pmatrix} A & 0 & B \\  Q & -I &  S \\  S^\top & 0 & R\end{pmatrix}.
    \end{align*}
    As $\Pi^{-1}$ has no poles on $\partial \mathbb D$ we know that the matrix pencil $\lambda E_{\Pi^{-1}} - A_{\Pi^{-1}}$ has no generalized eigenvalues on $\partial \mathbb D$.
    It is well known (see, e.g., \cite{Ionescu1992}), that the eigenvalues of $A+BK$ for any $K=-(B^\top Z B + R)^{-1}(A^\top Z B + S)^\top$ with $Z\in \texttt{dare}(A,B,Q,R,S)$ are a subset of the generalized eigenvalues of this pencil $\lambda E_{\Pi^{-1}} - A_{\Pi^{-1}}$.
    Therefore, we conclude that neither $\lambda(A+BK_\mathrm{s})$ nor $\lambda(A+BK_\mathrm{u})$ have eigenvalues on $\partial \mathbb{D}$.
\end{proof}

\begin{lemma}\label{lem:dare_zero}
  For $\filt_i=\ssrep{\A i}{\BB i}{\CC i}{\DD i}\in\RHinf^{m_i\times k_i}$, $i\in\{1,2\}$, $\A 1\in\R^{n_1\times n_1}$, $\A2 \in \R^{n_2\times n_2}$ and $\M\in\R^{m_1 \times m_2}$, assume that for all $\zs \in \partial\mathbb{D}$ it holds that $\filt^*_1(\zs)\M \filt_2(\zs) = 0$ and that $(\A 1,\BB 1)$ and $(\A 2,\BB 2)$ are controllable. Then there exists a matrix $Z_{12}\in\R^{n_1\times n_2}$ such that
  \begin{align}\label{eq:lem1}
    \begin{pmatrix} \A1 ^\top Z_{12} \A2 - Z_{12} & \A1^\top Z_{12} \BB2 \\ \BB1^\top Z_{12} \A2 & \BB1^\top Z_{12} \BB2 \end{pmatrix} = \begin{pmatrix}\CC1 & \DD1\end{pmatrix}^\top \M \begin{pmatrix}\CC2 & \DD2\end{pmatrix}.
  \end{align}
\end{lemma}
\begin{proof}
  The upper left block of~\eqref{eq:lem1} is the generalized Stein equation $\A1 ^\top Z_{12} \A2 - Z_{12} =\CC1^\top \M \CC2 $, which has a unique solution since both $A_1$ and $A_2$ are stable (compare \cite{Chu1987}).
  Due to $\lambda(\A i) \in\mathbb{D}$, it is straightforward to verify that this solution is given by 
  \begin{align}\label{eq:Z12}
    Z_{12} \triangleq -\sum_{k=0}^\infty \A1^{k\top} \CC1^\top \M \CC2 \A2^k.
  \end{align}
  Further, we observe for all $\zs\in\partial\mathbb D$ that
  \begin{align} 
    0&=\filt^*_1(\zs)\M \filt_2(\zs)= \underbrace{\BB1^\top (\zs I-\A1)^{-*} \CC1^\top \M \CC2 (\zs I-\A2)^{-1}\B2{}}_{\triangleq\,S(\zs)}  + \BB1^\top (\zs I-\A1)^{-*} \CC1^\top \M \DD2  + \DD1^\top\M \CC2 (\zs I-\A2)^{-1}\B2{} + \DD1^\top \M \DD2. \label{eq:tf_pi}
  \end{align}
  Due to $\lambda(\A i) \in\mathbb{D}$ for $i\in\{1,2\}$ and $|\zs| = 1$ the series
  \begin{align}\label{eq:(zI-A)_inv_series}
      (\zs I - \A 1)^{-*} = \sum_{n=1}^\infty \zs[n] \A 1 ^{(n-1) \top} \qquad \text{and} \qquad 
      (\zs I - \A 2)^{-1} = \sum_{n=1}^\infty \zs[-n] \A 2 ^{n-1} 
  \end{align}
  converge absolutely.
  Hence, we can rewrite the product $S(\zs)$ in~\eqref{eq:tf_pi} as
  \begin{align}\nonumber
    S(\zs)&= \sum_{k=0}^\infty \BB1^\top \A1^{k\top} \CC1^\top \M \CC2 \A2^k \BB2 + \sum_{n=1}^\infty \zs[n] \sum_{k=0}^\infty\BB1^\top \A1^{(k+n)\top} \CC1^\top \M \CC2 \A2^k \BB2 + \sum_{n=1}^\infty \zs[-n] \sum_{k=0}^\infty\BB1^\top \A1^{k\top} \CC1^\top \M \CC2 \A2^{k+n} \BB2\\
    &= -\BB1^\top Z_{12} \BB2 - \sum_{n=1}^\infty \left(\zs[n] \BB1^\top \A1^{n\top} Z_{12} \BB2 + \zs[-n] \BB1^\top Z_{12} \A2^n \BB2\right). \label{eq:Z12product}
  \end{align}
  As next step, we plug~\eqref{eq:(zI-A)_inv_series} and~\eqref{eq:Z12product} in~\eqref{eq:tf_pi}. Since~\eqref{eq:tf_pi} holds for all $\zs \in \partial \mathbb{D}$ and the polynomials $(\zs\mapsto \zs[n])_{n\in \mathbb Z}$ are linearly independent, we obtain by comparison of the coefficients the following equations
  \begin{subequations}
    \begin{align}
      0&=-\BB1^\top Z_{12} \BB2 + \DD1^\top M \DD2 \label{eq:lemeq1} \\
      0&=-\BB1^\top \A1^{n\top} Z_{12} \BB2 + \BB1^\top \A1^{(n-1)\top}\CC1 M \DD2 && \forall n\geq 1 \label{eq:lemeq2}\\
      0&=-\BB1^\top Z_{12} \A2^{n} \BB2 + \DD1 M \CC2 \A2^{n-1} \BB2 && \forall n\geq 1.\label{eq:lemeq3}
    \end{align}
  \end{subequations}
  Equation~\eqref{eq:lemeq1} verifies the lower right block of~\eqref{eq:lem1}.
  Given controllability of $(\A i,\BB i)$ for $i\in\{1,2\}$, the controllability matrix $S_i = \begin{pmatrix}\BB i & \A i \BB i & \dots & \A i^{n_i-1} \BB i \end{pmatrix}$
  has full rank, and thus a right inverse $ S_i S_i^\dagger = I$ exists. 
  Next, we stack~\eqref{eq:lemeq2} vertically for $n=1,\dots,n_1$ and obtain
  \begin{align*}
    0 = S_1^\top (-\A1^\top Z_{12} \BB2 + \CC1^\top M \DD2)
  \end{align*}
  which after left multiplication of $S_1^{\dagger \top}$ proves the upper right block of~\eqref{eq:lem1}.
  Finally, we stack~\eqref{eq:lemeq3} horizontally for $n=1,\dots,n_2$ and obtain
  \begin{align*}
    0 = (-\BB1^\top Z_{12} \A2 + \DD1^\top M \CC2)S_2
  \end{align*}
  which after right multiplication of $S_2^{\dagger}$ proves the lower left block of~\eqref{eq:lem1}.
\end{proof}

    \begin{lemma}\label{lem:convex_relaxation}
      Decompose $\hat X = \Zdone-\Zdtwo$ with $\Zdone\succeq 0$ and $\Zdtwo\succeq 0$. For any $\ZX\in\R^{\nfilthat \times m}$ and $\ZY\in \R^{\nfilthat\times m}$ we have
      \begin{align}\label{eq:convex_relaxation}
        \ZX^\top{\hat X} \ZX &\preceq
        \begin{pmatrix} \ZY \\ \ZX \end{pmatrix}^\top \begin{pmatrix} \Zdtwo & -\Zdtwo \\ -\Zdtwo & \Zdone\end{pmatrix} \begin{pmatrix} \ZY \\ \ZX \end{pmatrix} 
      \end{align}
      and the inequality holds with '$=$' if $\ZY = \ZX$ or $\Zdtwo =0$.
    \end{lemma}
    \begin{proof}
      First, the decomposition $\hat X = \Zdone-\Zdtwo$ with $\Zdone\succeq 0$ and $\Zdtwo\succeq 0$ exists because ${\hat X}={\hat X}^\top$ and thus such a decomposition can be obtained using the eigenvalue decomposition. 
      Next, we use $\ZX^\top {\hat X} \ZX = \ZX^\top  \Zdone \ZX - \ZX^\top \Zdtwo \ZX$
      and
      \begin{align*}
        0&\preceq ( \ZX-\ZY)^\top \Zdtwo( \ZX-\ZY)=\ZX^\top \Zdtwo \ZX + \ZY^\top \Zdtwo \ZY - \ZX^\top \Zdtwo \ZY - \ZY^\top \Zdtwo \ZX
      \end{align*}
      to deduce~\eqref{eq:convex_relaxation}. This inequality is tight if $\ZY = \ZX$ or $\Zdtwo = 0$.
    \end{proof}
\end{document}